\documentclass[runningheads]{llncs}
\usepackage[dvipsnames]{xcolor}
\definecolor{red}{HTML}{f94144}
\definecolor{orange1}{HTML}{f3722c}
\definecolor{orange2}{HTML}{f8961e}
\definecolor{yellow}{HTML}{f9c74f}
\definecolor{green1}{HTML}{90be6d}
\definecolor{green2}{HTML}{43aa8b}
\definecolor{blue1}{HTML}{577590}
\definecolor{blue2}{HTML}{197db4}
\colorlet{comment}{blue}

\usepackage{tikz}
\usetikzlibrary{positioning,fit,chains,shapes.symbols,calc}
\usepackage{hyperref}
\usepackage{rotating}
\usepackage{makecell}
\usepackage{caption}
\usepackage{subcaption}
\usepackage{csquotes}
\usepackage{enumitem}
\usepackage{graphicx}
\usepackage{mathtools}
\usepackage{booktabs}
\usepackage{makecell}
\usepackage{multirow}
\usepackage[ruled,lined,vlined]{algorithm2e}

\SetCommentSty{xCommentSty}

\begin{document}
\title{Visualizing Trace Variants From Partially Ordered Event Data}
\author{Daniel Schuster\inst{1,2} \orcidID{0000-0002-6512-9580} \and
Lukas Schade\inst{2} \and
Sebastiaan J. van Zelst\inst{1,2} \orcidID{0000-0003-0415-1036} \and
Wil M. P. van der Aalst\inst{1,2} \orcidID{0000-0002-0955-6940} }

\authorrunning{D. Schuster et al.}
\institute{Fraunhofer Institute for Applied Information Technology FIT, Germany\\
\email{\{daniel.schuster,sebastiaan.van.zelst\}@fit.fraunhofer.de}
\and
RWTH Aachen University, Aachen, Germany\\
\email{wvdaalst@pads.rwth-aachen.de}}

\maketitle              % typeset the header of the contribution
\begin{abstract}
Executing operational processes generates event data, which contain information on the executed process activities. % and are stored within organizations' information systems. 
Process mining techniques allow to systematically analyze event data to gain insights that are then used to optimize processes.
Visual analytics for event data are essential for the application of process mining. %, e.g., to explore the event data, to find patterns, and to decide on further process mining techniques to be applied.
Visualizing unique process executions---also called trace variants, i.e., unique sequences of executed process activities---is a common technique implemented in many scientific and industrial process mining applications. 
Most existing visualizations assume a total order on the executed process activities, i.e., these techniques assume that process activities are atomic and were executed at a specific point in time.
In reality, however, the executions of activities are \emph{not} atomic. 
Multiple timestamps are recorded for an executed process activity, e.g., a start-timestamp and a complete-timestamp.
Therefore, the execution of process activities may overlap and, thus, cannot be represented as a total order if more than one timestamp is to be considered.
In this paper, we present a visualization approach for trace variants that incorporates start- and complete-timestamps of activities.
%The proposed visualization approach has been implemented in Cortado, an interactive process modeling and discovery tool. 

\end{abstract}

\section{Introduction}

The execution of operational processes, e.g., business and production processes, is often supported by information systems that record process executions in detail.
We refer to such recorded information as \emph{event data}.
The analysis of event data is of great importance for organizations to improve their processes.
\emph{Process mining}~\cite{vanderAalst2016} offers various techniques for systematically analyzing event data, e.g., to learn a process model, to check compliance, and to obtain performance measures.
These insights into the processes can then be used to optimize them.

\iffalse
\begin{figure}[tb]
    \centering
    \includegraphics[width=.9\textwidth]{images/variant-explorer-example.png}
    \caption{Example of a classical trace variant visualization, i.e., a variant explorer. For example, the first row shows that in the given event log 56,482 process executions followed the pattern that first the activity \enquote{Create Fine} was executed, followed by \enquote{Send Fine}, \enquote{Insert Fine Notification}, \enquote{Add penalty}, and \enquote{Send for Credit Collection}}
    \label{fig:variant-explorer-classic}
\end{figure}
\fi

\begin{figure}
    \centering
    \begin{subfigure}[b]{\textwidth}
        \centering
        \begin{tikzpicture}
        \scriptsize
          \tikzset{
            arrow/.style={
              draw=lightgray,
              minimum height=.5cm,
              inner sep=.2em,
              shape=signal,
              signal from=west,
              signal to=east,
              signal pointer angle=150,
            }
          }
          %\node[arrow] {first};
          \node[arrow,align=center,fill=red,text=white] (1) {activity $A$};
          \node[arrow,align=center,fill=orange, right=.1cm of 1] (2) {activity $B$};
          \node[arrow,align=center,fill=green1, right=.1cm of 2] (3) {activity $C$};
          \node[arrow,align=center,fill=blue1,text=white, right=.1cm of 3] (4) {activity $D$};
          \node[arrow,align=center,fill=yellow, right=.1cm of 4] (5) {activity $E$};
          \node[arrow,align=center,fill=green2, right=.1cm of 5] (6) {activity $F$};
          \node[arrow,align=center,fill=red,text=white, right=.1cm of 6] (7) {activity $A$};
          \node[arrow,align=center,fill=blue2,text=white, right=.1cm of 7] (8) {activity $G$};
          
          \node[arrow,align=center,fill=red,text=white, below=.2cm of 1] (11) {activity $A$};
          \node[arrow,align=center,fill=orange, right=.1cm of 11] (21) {activity $B$};
          \node[arrow,align=center,fill=green1, right=.1cm of 21] (31) {activity $C$};
          \node[arrow,align=center,fill=yellow, right=.1cm of 31] (41) {activity $E$};
          \node[arrow,align=center,fill=blue1,text=white, right=.1cm of 41] (51) {activity $D$};
          \node[arrow,align=center,fill=red,text=white, right=.1cm of 51] (61) {activity $A$};
          \node[arrow,align=center,fill=green2, right=.1cm of 61] (71) {activity $F$};
          \node[arrow,align=center,fill=blue2,text=white, right=.1cm of 71] (81) {activity $G$};
          
            \node[arrow,align=center,fill=red,text=white, below=.2cm of 11] (12) {activity $A$};
          \node[arrow,align=center,fill=orange, right=.1cm of 12] (22) {activity $B$};
          \node[arrow,align=center,fill=green1, right=.1cm of 22] (32) {activity $C$};
          \node[arrow,align=center,fill=yellow, right=.1cm of 32] (42) {activity $E$};
          \node[arrow,align=center,fill=blue1,text=white, right=.1cm of 42] (52) {activity $D$};
          %\node[arrow,align=center,fill=black!5, right=.1cm of 52] (62) {activity $A$};
          %\node[arrow,align=center,fill=black!5, right=.1cm of 62] (72) {activity $F$};
          \node[arrow,align=center,fill=blue2,text=white, right=.1cm of 52] (82) {activity $G$};
          
        \end{tikzpicture}
        \caption{Three different trace variants showing the execution order of \emph{atomic} activities}
        \label{fig:classic_variant_explorer_a}
    \end{subfigure}
    \hfill
\begin{subfigure}[b]{\textwidth}
        \centering
        \begin{tikzpicture}
        \scriptsize
          \tikzset{
            arrow/.style={
              draw=lightgray,
              minimum height=.5cm,
              inner sep=.2em,
              shape=signal,
              signal from=west,
              signal to=east,
              signal pointer angle=150,
            }
          }
          %\node[arrow] {first};
          \node[arrow,align=center,fill=red,text=white] (1) {activity $A$\\(start)};
          \node[arrow,align=center,fill=orange, right=.1cm of 1] (2) {activity $B$\\(start)};
          \node[arrow,align=center,fill=green1, right=.1cm of 2] (3) {activity $C$\\(start)};
          \node[arrow,align=center,fill=red,text=white, right=.1cm of 3] (4) {activity $A$\\(complete)};
          \node[arrow,align=center,fill=orange, right=.1cm of 4] (5) {activity $B$\\(complete)};
          \node[arrow,align=center,fill=blue1,text=white, right=.1cm of 5] (6) {activity $D$\\(start)};
          \node[arrow,align=center,fill=yellow, right=.1cm of 6] (7) {activity $E$\\(start)};
          \node[align=center,right=.1cm of 7] (8) {$\dots$};
          
          \node[arrow,align=center,fill=red,text=white, below=.2cm of 1] (11) {activity $A$\\(start)};
          \node[arrow,align=center,fill=orange, right=.1cm of 11] (21) {activity $B$\\(start)};
          \node[arrow,align=center,fill=green1, right=.1cm of 21] (31) {activity $C$\\(start)};
          \node[arrow,align=center,fill=red,text=white, right=.1cm of 31] (41) {activity $A$\\(complete)};
          \node[arrow,align=center,fill=orange, right=.1cm of 41] (51) {activity $B$\\(complete)};
          \node[arrow,align=center,fill=yellow, right=.1cm of 51] (61) {activity $E$\\(start)};
          \node[arrow,align=center,fill=blue1,text=white, right=.1cm of 61] (71) {activity $D$\\(start)};
          \node[align=center,right=.1cm of 71] (81) {$\dots$};
          
          \iffalse
            \node[arrow,align=center,fill=red,text=white, below=.2cm of 11] (12) {activity $A$};
          \node[arrow,align=center,fill=orange, right=.1cm of 12] (22) {activity $B$};
          \node[arrow,align=center,fill=green1, right=.1cm of 22] (32) {activity $C$};
          \node[arrow,align=center,fill=yellow, right=.1cm of 32] (42) {activity $E$};
          \node[arrow,align=center,fill=blue1,text=white, right=.1cm of 42] (52) {activity $D$};
          %\node[arrow,align=center,fill=black!5, right=.1cm of 52] (62) {activity $A$};
          %\node[arrow,align=center,fill=black!5, right=.1cm of 62] (72) {activity $F$};
          \node[arrow,align=center,fill=blue2,text=white, right=.1cm of 52] (82) {activity $G$};
          \fi
          
        \end{tikzpicture}
        \caption{Two different trace variants showing the execution order of \emph{non-atomic} activities, i.e, each activity is split into start and complete}
        \label{fig:classic_variant_explorer_b}
    \end{subfigure}
    \caption{Classic trace variant visualizations for \emph{(non)-atomic} process activities}
    \label{fig:classic_variant_explorer}
\end{figure}
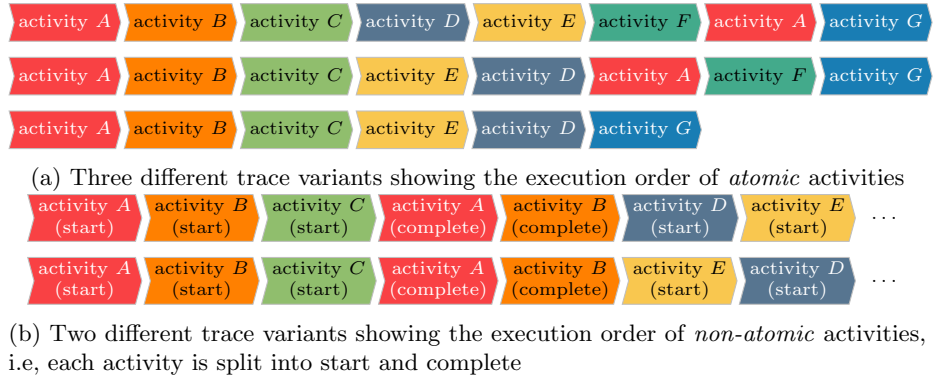

As in other data analysis applications, \emph{visual analytics} for event data are important in the application of process mining.
%The importance of said techniques is supported by a 
A state-of-the-art process mining methodology~\cite{10.1007/978-3-319-19069-3_19} lists \emph{process analytics} including visual analytics as a key component next to the classic fields of process mining: process discovery, conformance checking, and process enhancement.

%\vspace{-1cm}
A visualization approach that is used across various process mining tools, ranging from industry to scientific tools, is called the \emph{variant explorer}.
Consider \autoref{fig:classic_variant_explorer_a} for an example.
In classic trace variant visualizations, a \emph{variant} describes a unique sequence of executed process activities. 
%An \emph{event log}, a collection of event data from a specific process, contains usually various different variants.
Thus, a \emph{strict total order} on the contained activities is required to visualize such sequence.
Recorded timestamps of the executed activities are usually used for ordering them.

\begin{figure}
    \centering
    \begin{subfigure}[b]{0.65\textwidth}
        \resizebox{\textwidth}{!}{%
        \begin{tikzpicture}
            \tikzstyle{labels}=[font=\footnotesize,color=gray]
            \node[align=center](text) at (4.5,3.1) {\emph{Case/Process Instance 1}};
            \node[align=center](text2) at (4.5,3.5) {};
            
            \draw[dashed,lightgray] (0,.5) to (0,3);
            \node[labels] at (0,.3) {08:00};
            \draw[dashed,lightgray] (1,.5) to (1,3);
            \node[labels] at (1,.3) {09:00};
            \draw[dashed,lightgray] (2,.5) to (2,3);
            \node[labels] at (2,.3) {10:00};
            \draw[dashed,lightgray] (3,.5) to (3,3);
            \node[labels] at (3,.3) {11:00};
            \draw[dashed,lightgray] (4,.5) to (4,3);
            \node[labels] at (4,.3) {12:00};
            \draw[dashed,lightgray] (5,.5) to (5,3);
            \node[labels] at (5,.3) {13:00};
            \draw[dashed,lightgray] (6,.5) to (6,3);
            \node[labels] at (6,.3) {14:00};
            \draw[dashed,lightgray] (7,.5) to (7,3);
            \node[labels] at (7,.3) {15:00};
            \draw[dashed,lightgray] (8,.5) to (8,3);
            \node[labels] at (8,.3) {16:00};
            \draw[dashed,lightgray] (9,.5) to (9,3);
            \node[labels] at (9,.3) {17:00};
        
            \node[circle,fill=black,inner sep=0pt, minimum size=3pt] (A+S) at (0,2.6) {};
            \node[circle,fill=black,inner sep=0pt, minimum size=3pt] (A+C) at (1.5,2.6) {};
            \draw[-] (A+S) to node[above] {A} (A+C);
            
            \node[circle,fill=black,inner sep=0pt, minimum size=3pt] (B+S) at (.5,2.1) {};
            \node[circle,fill=black,inner sep=0pt, minimum size=3pt] (B+C) at (3,2.1) {};
            \draw[-] (B+S) to node[above] {B} (B+C);
            
            \node[circle,fill=black,inner sep=0pt, minimum size=3pt] (C+S) at (1,1.6) {};
            \node[circle,fill=black,inner sep=0pt, minimum size=3pt] (C+C) at (4,1.6) {};
            \draw[-] (C+S) to node[above] {C} (C+C);
            
            \node[circle,fill=black,inner sep=0pt, minimum size=3pt] (D+S) at (3.5,1.1) {};
            \node[circle,fill=black,inner sep=0pt, minimum size=3pt] (D+C) at (5.5,1.1) {};
            \draw[-] (D+S) to node[above] {D} (D+C);
            
            \node[circle,fill=black,inner sep=0pt, minimum size=3pt] (E+S) at (3.7,.6) {};
            \node[circle,fill=black,inner sep=0pt, minimum size=3pt] (E+C) at (5,.6) {};
            \draw[-] (E+S) to node[above] {E} (E+C);
            
            \node[circle,fill=black,inner sep=0pt, minimum size=3pt] (F+S) at (6,2.6) {};
            \node[circle,fill=black,inner sep=0pt, minimum size=3pt] (F+C) at (7,2.6) {};
            \draw[-] (F+S) to node[above] {F} (F+C);
            
            \node[circle,fill=black,inner sep=0pt, minimum size=3pt] (G+S) at (6.5,2.1) {};
            \node[circle,fill=black,inner sep=0pt, minimum size=3pt] (G+C) at (8,2.1) {};
            \draw[-] (G+S) to node[above] {A} (G+C);
            
            \node[circle,fill=black,inner sep=0pt, minimum size=3pt] (H+S) at (8.5,2.6) {};
            \node[circle,fill=black,inner sep=0pt, minimum size=3pt] (H+C) at (9,2.6) {};
            \draw[-] (H+S) to node[above] {G} (H+C);
            
            \draw[->] (0,.5) to (9,.5);
            \node[] (t) at (9,.7) {$t$};
        
        \end{tikzpicture}
        }
        \resizebox{\textwidth}{!}{%
        \begin{tikzpicture}
            \tikzstyle{labels}=[font=\footnotesize,color=gray]
            \node[align=center](text) at (4.5,3.1) {\emph{Case/Process Instance 2}};
            \node[align=center](text2) at (4.5,3.5) {};
            
            \draw[dashed,lightgray] (0,.5) to (0,3);
            \node[labels] at (0,.3) {08:00};
            \draw[dashed,lightgray] (1,.5) to (1,3);
            \node[labels] at (1,.3) {09:00};
            \draw[dashed,lightgray] (2,.5) to (2,3);
            \node[labels] at (2,.3) {10:00};
            \draw[dashed,lightgray] (3,.5) to (3,3);
            \node[labels] at (3,.3) {11:00};
            \draw[dashed,lightgray] (4,.5) to (4,3);
            \node[labels] at (4,.3) {12:00};
            \draw[dashed,lightgray] (5,.5) to (5,3);
            \node[labels] at (5,.3) {13:00};
            \draw[dashed,lightgray] (6,.5) to (6,3);
            \node[labels] at (6,.3) {14:00};
            \draw[dashed,lightgray] (7,.5) to (7,3);
            \node[labels] at (7,.3) {15:00};
            \draw[dashed,lightgray] (8,.5) to (8,3);
            \node[labels] at (8,.3) {16:00};
            \draw[dashed,lightgray] (9,.5) to (9,3);
            \node[labels] at (9,.3) {17:00};
            
            \node[circle,fill=black,inner sep=0pt, minimum size=3pt] (A+S) at (0,2.6) {};
            \node[circle,fill=black,inner sep=0pt, minimum size=3pt] (A+C) at (1.5,2.6) {};
            \draw[-] (A+S) to node[above] {A} (A+C);
            
            \node[circle,fill=black,inner sep=0pt, minimum size=3pt] (B+S) at (1,2.1) {};
            \node[circle,fill=black,inner sep=0pt, minimum size=3pt] (B+C) at (2,2.1) {};
            \draw[-] (B+S) to node[above] {B} (B+C);
            
            \node[circle,fill=black,inner sep=0pt, minimum size=3pt] (C+S) at (1,1.6) {};
            \node[circle,fill=black,inner sep=0pt, minimum size=3pt] (C+C) at (4,1.6) {};
            \draw[-] (C+S) to node[above] {C} (C+C);
            
            \node[circle,fill=black,inner sep=0pt, minimum size=3pt] (D+S) at (3.5,1.1) {};
            \node[circle,fill=black,inner sep=0pt, minimum size=3pt] (D+C) at (5.5,1.1) {};
            \draw[-] (D+S) to node[above] {D} (D+C);
            
            \node[circle,fill=black,inner sep=0pt, minimum size=3pt] (E+S) at (3,.6) {};
            \node[circle,fill=black,inner sep=0pt, minimum size=3pt] (E+C) at (5,.6) {};
            \draw[-] (E+S) to node[above] {E} (E+C);
            
            \node[circle,fill=black,inner sep=0pt, minimum size=3pt] (F+S) at (6,2.6) {};
            \node[circle,fill=black,inner sep=0pt, minimum size=3pt] (F+C) at (7,2.6) {};
            \draw[-] (F+S) to node[above] {F} (F+C);
            
            \node[circle,fill=black,inner sep=0pt, minimum size=3pt] (G+S) at (5.8,2.1) {};
            \node[circle,fill=black,inner sep=0pt, minimum size=3pt] (G+C) at (7,2.1) {};
            \draw[-] (G+S) to node[above] {A} (G+C);
            
            \node[circle,fill=black,inner sep=0pt, minimum size=3pt] (H+S) at (7.5,2.6) {};
            \node[circle,fill=black,inner sep=0pt, minimum size=3pt] (H+C) at (9,2.6) {};
            \draw[-] (H+S) to node[above] {G} (H+C);
            
            \draw[->] (0,.5) to (9,.5);
            \node[] (t) at (9,.7) {$t$};
        
        \end{tikzpicture}
        }
        \caption{Time plots visualizing activity instances, i.e., each activity has a start-timestamp and a complete-timestamp, executed within two different cases/process instance}
        \label{fig:time-plot-activity-instances}
    \end{subfigure}
    \hfill
    \begin{subfigure}[b]{0.3\textwidth}
        \begin{tikzpicture}
            \node[] (A) {A};
            \node[below of=A] (B) {B};
            \node[below of=B] (C) {C};
            \node[right of= A] (D) {D};
            \node[right of= B] (E) {E};
            
            \draw[->] (A) -- (D);
            \draw[->] (A) -- (E);
            \draw[->] (B) -- (D);
            \draw[->] (B) -- (E);
            
            \node[right of= D] (F) {F};
            \node[right of= E] (G) {A};        
            
            \draw[->] (D) -- (F);
            \draw[->] (E) -- (G);
            \draw[->] (D) -- (G);
            \draw[->] (E) -- (F);
            \draw[->, bend right=15] (C) to (F);
            \draw[->] (C) -- (G);
            \node[right of= F] (H) {G};     
            
            \draw[->] (F) -- (H);
            \draw[->] (G) -- (H);
            
            \draw[->, bend left, dotted] (A) to (F);
            \draw[->, bend left=40, dotted] (A) to (H);
            \draw[->, bend left, dotted] (D) to (H);
            
            \draw[->,dotted] (E) to (H);
            
            \draw[->, bend right=20, dotted] (B) to (G);
            \draw[->, bend right=55, dotted] (B) to (H);
            
            \draw[->, bend right=40, dotted] (C) to (H);
            \draw[->,dotted] (B) -- (F);
            \draw[->,dotted] (B) -- (H);
    
        \end{tikzpicture}
        \caption{Visualization of the corresponding interval order. Vertices represent activity instances. Arcs indicate an ordering between activity instances}
        \label{fig:partial-order}
    \end{subfigure}
    \caption{Visualizing partially ordered event data. Each interval shown in \autoref{fig:time-plot-activity-instances}, i.e., an activity instance, describes the execution of an activity. $A,\dots,G$ represent activity labels. Both visualized cases/process instances (\autoref{fig:time-plot-activity-instances}) correspond to the same interval order (\autoref{fig:partial-order}). Note that we consider two activity instances to be unrelated if they overlap in time}
    \label{fig:example_instance_and_partial_order}
\end{figure}
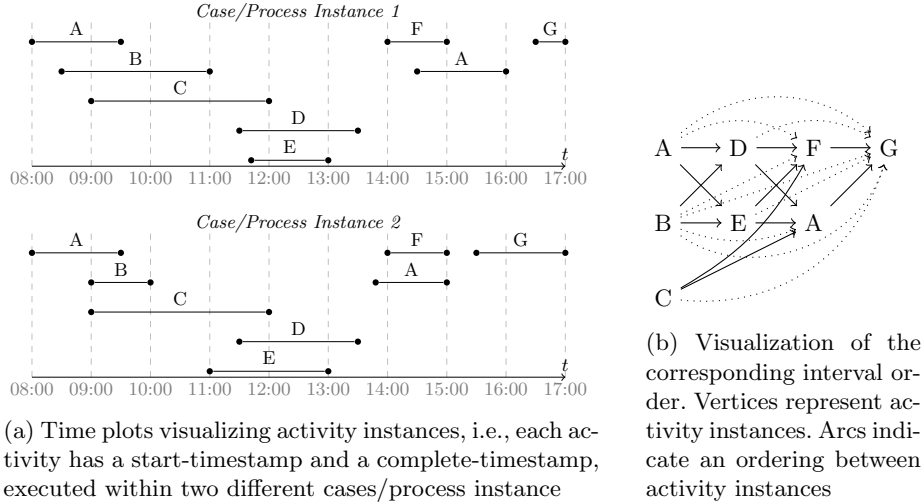

This classic trace variant visualization has two main limitations.
(1) Assume atomic process activities, i.e., a single timestamp is recorded for each process activity.
A \emph{strict} total order cannot be derived if multiple activities have the same timestamp. 
In such cases, the sequential visualization, indicating temporal execution, of process activities is problematic because a second-order criteria is needed to obtain a strict total order.
(2) In many real-life scenarios, process activities are performed over time, i.e., they are \emph{non}-atomic.
Thus, the execution of activities may intersect with each other.
Consider \autoref{fig:time-plot-activity-instances} for an example.
Considering both start- and complete-timestamps, a strict total order cannot be obtained if the executions of activities overlap.
The classic trace variant explorer usually splits the activities in start and complete as shown in \autoref{fig:classic_variant_explorer_b} to obtain atomic activities.
However, the parallel behavior of activities is not easily discernible from the visualization.
In addition, the first limitation remains.

In this paper, we propose a novel visualization of trace variants to overcome the two aforementioned limitations.  
We define a variant as an \emph{interval order}, which can be represented as a graph.
For instance, \autoref{fig:partial-order} shows the interval order of the two process executions shown in \autoref{fig:time-plot-activity-instances}.
%In the example, the activities $A$ and $B$ are unrelated to each other (no path from $A$ to $B$ and vice versa exists) since the executions of these activities overlap (\autoref{fig:time-plot-activity-instances}), but both are executed before $D$ and $E$ (indicated by arcs) and other activities executed later.
The graph representation of an interval order (cf. \autoref{fig:partial-order}) is, however, not easy to read compared to the classic trace variant explorer (cf. \autoref{fig:classic_variant_explorer}).
Therefore, we propose an approach to derive a visualization from interval orders representing trace variants.

The remainder of this paper is structured as follows.
\autoref{sec:related_work} presents related work. 
\autoref{sec:background} introduces concepts and definitions used throughout this paper.
\autoref{sec:approach} introduces the proposed visualization approach.
\autoref{sec:evaluation} presents an experimental evaluation, and \autoref{sec:conclusion} concludes this paper.

\section{Related Work}
\label{sec:related_work}
%This section presents related work.
%We focus here on related work tailored for partially ordered event data and state-of-the-art visualization approaches for trace variants.

For a general overview of process mining, we refer to~\cite{vanderAalst2016}.
Note that the majority of process mining techniques assume totally ordered event data.
For example, in process discovery few algorithms exist that utilize life cycle information, i.e., more than one timestamp, of the recorded process activities.
For instance, the Inductive Miner algorithm has been extended in~\cite{10.1007/978-3-319-42887-1_17} to utilize start- and complete-timestamps of process activities.
Also in conformance checking there exist algorithms that utilize life cycle information, e.g., \cite{10.1007/978-3-319-15895-2_7}. 
A complete overview of techniques utilizing life cycle information is outside the scope of this paper. 

In~\cite{10.1007/978-3-319-19069-3_19}, the authors present a methodology for conducting process mining projects and highlight the importance of visual analytics.
In~\cite{10.1007/978-3-319-53435-0_7}, open challenges regarding visual analytics in process mining are presented.
The visualization of time-oriented event data---the topic of this paper---is identified as a challenge.

The classic variant explorer as shown in \autoref{fig:classic_variant_explorer} can be found in many different process mining tools, e.g., in ProM\footnote{\url{https://www.promtools.org}}, which is an open-source process mining software tool. 
In~\cite{10.2312:eurova.20151106}, the authors present a software tool to visualize event data.
Various visualizations of event data are offered; however, a variant explorer, as considered in this paper, is not available.
In \cite{van2021mayItakeYourOrder}, the authors present a plugin for ProM to visualize partially ordered event data.
The approach considers events to be atomic, i.e., an event representing the start and an event representing the completion of an activity are considered to be separate events.
Based on a user-selected time granularity, events within the same time segment are aggregated, i.e., they are considered and visualized to be executed in parallel.
This offers the advantage that the user can change the visualization depending on how accurately the timestamps are to be interpreted.
Compared to our approach, we consider non-atomic activity instances, i.e., we map start and complete events of a process activity to an activity instance.
Next, we relate these activity instances to each other instead of atomic events as proposed in~\cite{van2021mayItakeYourOrder}.
Therefore, both approaches, the one presented in~\cite{van2021mayItakeYourOrder} and the one presented in this paper, can coexist and each have their advantages and disadvantages.

\section{Preliminaries}
\label{sec:background}
In this section, we present concepts and definitions used within this paper.

\begin{table}[tb]
\caption{Example of event data}
\label{tab:event_data}
\centering
\scriptsize
\begin{tabular}{ c  c  c  c  c  c  c }
    \toprule
    \textbf{Event-ID} & \textbf{Case-ID} & \textbf{Activity Label} & \textbf{Start-timestamp} & \textbf{Complete-timestamp} & \textbf{Resource} & \dots \\
    \midrule
    %\vdots & \vdots & \vdots & \vdots & \vdots & \vdots \\
    1 & 1 & activity $A$ & 07/13/2021 08:00 & 07/13/2021 09:30 & staff & \dots \\  
    2 & 1 & activity $B$ & 07/13/2021 08:30 & 07/13/2021 11:00 & staff &\dots \\  
    3 & 1 & activity $C$ & 07/13/2021 09:00 & 07/13/2021 12:00 & staff & \dots \\  
    4 & 1 & activity $D$ & 07/13/2021 11:30 & 07/13/2021 13:30 & staff & \dots \\  
    5 & 1 & activity $E$ & 07/13/2021 11:40 & 07/13/2021 13:00 & supervisor & \dots \\  
    6 & 1 & activity $F$ & 07/13/2021 14:00 & 07/13/2021 15:00 & manager & \dots \\  
    7 & 1 & activity $A$ & 07/13/2021 14:30 & 07/13/2021 16:00 & staff & \dots \\  
    8 & 1 & activity $G$ & 07/13/2021 16:30 & 07/13/2021 17:00 & staff & \dots\\
    9 & 2 & activity $A$ & 07/13/2021 08:00 & 07/13/2021 09:30 & staff & \dots \\  
    %10 & 2 & activity $B$ & 07/13/2021 09:00 & 07/13/2021 10:00 & staff & \dots \\  
   % 11 & 2 & activity $C$ & 07/13/2021 09:00 & 07/13/2021 12:00 & staff & \dots \\  
    \vdots & \vdots & \vdots & \vdots & \vdots & \vdots & \vdots \\\bottomrule
\end{tabular}
\end{table}

Event data describes the historical execution of processes.
\autoref{tab:event_data} shows an example of said event data.
Each row corresponds to an event, i.e., in the given example an \emph{activity instance}.\footnote{Note that in some event logs, the start and the completion of an activity are separate events (i.e., separate rows). Observe that such records are easily transformed to our notion of event data.}
For example, the first event, identified by event-id $1$, recorded that activity $A$ has been executed from 08:00 until 09:30 at 07/13/2021 within the process instance identified by case-id $1$.  

In general, activity instances describe the execution of a process activity within a specific case.
A \emph{case} describes a single execution of a process, i.e., a process instance, and it is formally a set of activity instances that have been executed for the same case. 
Activity instances consist of at least the following attributes: an identifier, a case-id, an activity label, a start-timestamp, and a complete-timestamp.
%As shown in \autoref{tab:event-log}, activity instances can have additional attributes, e.g., resource, cost, and .
Since we are only interested in the order of activity instances within a case and not in possible additional attributes of an activity instance, we define activity instances as a 5-tuple.
%For instance, the first 

\begin{definition}[Universes]
%Let $\mathcal{N}$ be the universe of attribute names. %with $\{t_{start},t_{complete},case,act\} {\subseteq} \mathcal{N}$ and 
%Let $\mathcal{V}$ be the universe of attribute values.
$\mathcal{T}$ is the universe of totally ordered timestamps. 
$\mathcal{L}$ is the universe of activity labels.
$\mathcal{C}$ is the universe of case identifiers.
$\mathcal{I}$ is the universe of activity instance identifiers.
\end{definition}

%Given the universes, we define an activity instance as a 5-tuple.

\begin{definition}[Activity Instance\iffalse\footnote{Note that this definition does not allow to distinguish two activity instances where the same activity is executed within the same case at the same timestamp. Additionally, activity instance ids would be required. For simplicity, however, we do not add ids to activity instances.}\fi]
An activity instance $(i,c,l,t_s,t_c) {\in} \allowbreak \mathcal{I} {\times} \mathcal{C} {\times} \allowbreak \mathcal{L} {\times} \allowbreak \mathcal{T} {\times} \mathcal{T}$ describes the execution of an activity labeled $l$ within the case $c$.  
The start-timestamp of the activity's execution is $t_s$, and the complete-timestamp is $t_c$, where $t_s {\leq} t_c$.
Each activity instance is uniquely identifiable by $i$.
We denote the universe of activity instances by $\mathcal{A}$.
\end{definition}

Note that any event log with only one timestamp per executed activity can also be easily expressed in terms of activity instances, i.e., $t_s{=}t_c$. 
For a given activity instance $a {=} (i,c,l,t_s,t_c) {\in} \mathcal{A}$, we define projection functions: $\pi^{i}(a) {=} i$, $\pi^{c}(a) {=} c$, $\pi^{l}(a) {=} l$, $\pi^{t_s}(a) {=} t_s$, and $\pi^{t_c}(a) {=} t_c$. 

\begin{definition}[Event Log]
An event log $E$ is a set of activity instances, i.e., $E {\subseteq} \mathcal{A}$ such that for $a_1,a_2{\in}E \land \pi^i(a_1){=}\pi^i(a_2) \Rightarrow a_1{=}a_2$.
We denote the universe of event logs by $\mathcal{E}$.
\end{definition}

For a given event log $E {\in} \mathcal{E}$, we refer to the set of activity instances executed within a given case $c{\in}\mathcal{C}$ as a \emph{trace}, i.e., $T_c{=}\{a {\in} E \mid \land \pi^c(a) {=} c\}$.
As shown in \autoref{fig:time-plot-activity-instances}, we can visualize a trace and its activity instances in a time plot.

Note that each activity instance $a {=} (i,c,l,t_s,t_c) {\in} \mathcal{A}$ defines an interval on the timeline, i.e., $[t_s,t_c]$. 
A collection of intervals---in this paper we focus on traces---defines an \emph{interval order}.
In general, given two activity instances $a_1,a_2{\in}\mathcal{A}$, we say $a_1{<}a_2$ iff $\pi^{t_c}(a_1){<}\pi^{t_s}(a_2)$.
Note that interval orders are a proper subclass of strict partial orders~\cite{FISHBURN1970144}; hence, interval orders satisfy: irreflexivity, transitivity, and asymmetry. 
Interval orders additionally satisfy the \emph{interval order condition}, i.e., for any $x,y,w,z : x{<}w \land y{<}z \Rightarrow x{<}z \lor y{<}w$~\cite{FISHBURN1970144}.
%In general, every interval order is an interval order but not vice versa.
%A further overview of order theory is outside the scope of this paper. 
%For more details, we refer to~\cite{BOUSQUETMELOU2010884,FISHBURN1970144}.

%For a given trace, we define its \emph{interval order}, indicating the ordering relations between the different executed activity instances.
In this paper, we represent an interval order as a directed, labeled graph that consists of vertices $V$, representing activity instances, and directed edges $V {\times} V$, representing ordering relations between activity instances.
\autoref{fig:partial-order} shows the interval order of the traces shown in \autoref{fig:time-plot-activity-instances}.
We observe that the first two activity instances labeled with $A$ and $B$ are incomparable to each other because there is no arc from either $A$ to $B$ or vice versa. 
Thus, the first execution of $A$ and $B$ are executed in parallel, i.e., their intervals overlap.
For example, activity $C$ is related to $F$, $G$ and the second execution of $A$.
Thus, $C$ is executed before $F$, $G$ and the second execution of $A$.
Next, we formally define the construction of the directed graph representing the interval order of a trace. 

\begin{definition}[Interval Order of a Trace]
Given a trace $T_c {\subseteq} \mathcal{A}$, we define the corresponding interval order as a labeled, directed graph $(V,E,\lambda)$ consisting of vertices $V$, directed edges $E{=}(V {\times} V)$, and a labeling function $\lambda:V {\to} \mathcal{L}$.
The set of vertices is defined by $V{=}  T_c $ with $\lambda(a){=}\pi^l(a)$. %for all $a {\in} T$.
Given two activity instances $a^1,a^2 {\in} T$, there is a directed edge $\big(\pi^i(a_1),\pi^i(a_2)\big) {\in} E$ iff $\pi^{t_c}(a_1) {<} \pi^{t_s}(a_2)$.
We denote the universe of interval orders by $\mathcal{P}$.
\end{definition}
Next, we define the induced interval order.
\begin{definition}[Induced Interval Order]
Given $(V,E,\lambda) {\in} \mathcal{P}$.
For $V' {\subseteq} V$, we define the induced interval order, i.e., the induced subgraph, $(V',E',\allowbreak\lambda') {\in} \mathcal{P}$ with $E' {=} E {\cap} (V' {\times} V')$ and $ \lambda'(v){=}\lambda(v)$ for all $v{\in}V'$.
\end{definition}

\section{Visualizing Trace Variants}
\label{sec:approach}
This section introduces the proposed approach to visualize trace variants from partially ordered event data.
%The section is structured as follows.
\autoref{sec:approach_subsection} introduces the approach, and \autoref{sec:guarantees} proves that the approach is deterministic.
\autoref{sec:limitations} discusses the potential limitations of the approach.
Finally, \autoref{sec:implementation} covers the implementation.

\subsection{Approach}
\label{sec:approach_subsection}

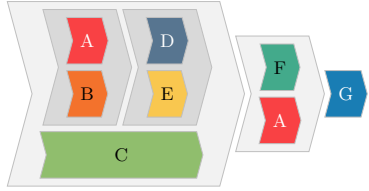
\begin{figure}[tb]
    \centering
    \centering
        \resizebox{.4\textwidth}{!}{%
        \begin{tikzpicture}
          \tikzset{
            arrow/.style={
              draw=lightgray,
              minimum height=.8cm,
              inner sep=.4em,
              shape=signal,
              signal from=west,
              signal to=east,
              signal pointer angle=150,
            }
          }
          %\node[arrow] {first};
          \node[arrow,align=center,fill=black!5] (1) 
            {\begin{tikzpicture}
                \node[arrow, fill=black!15] (AB) 
                    {\begin{tikzpicture}
                        \node[arrow,fill=red,text=white] (A) {A};
                        \node[arrow,below=.1cm of A,fill=orange1] (B) {B};
                    \end{tikzpicture}};
                \node[arrow, right=.1cm of AB,fill=black!15] (DE) 
                    {\begin{tikzpicture}
                        \node[arrow,fill=blue1,text=white] (D) {D};
                        \node[arrow,below=.1cm of D,fill=yellow] (E) {E};
                    \end{tikzpicture}};
                \coordinate (CENTER) at ($(AB)!0.5!(DE)$);
                \node[arrow,minimum width=2.8cm,xshift=-.1cm,below=1.1cm of CENTER,fill=green1] (C) {C};
            \end{tikzpicture}};
          
          \node[arrow,align=center,right=.1cm of 1,fill=black!5] (2)   
            {\begin{tikzpicture}
                \node[arrow,fill=green2] (F) {F};
                \node[arrow, below=.1cm of F,fill=red,text=white] (G) {A};
            \end{tikzpicture}};
          
          \node[arrow,align=center,right=.1cm of 2,fill=blue2,text=white] (3) {G};
        \end{tikzpicture}
        }
    \caption{Proposed visualization for the interval order shown in \autoref{fig:partial-order}}
    \label{fig:variant_visualization}
\end{figure}

The proposed visualization approach of trace variants is based on chevrons, a graphical element known from classical trace variant visualizations (cf. \autoref{fig:classic_variant_explorer}).
\autoref{fig:variant_visualization} shows an example of the proposed visualization for the interval order given in \autoref{fig:partial-order}.
The interpretation of a chevron as indicating sequential order is maintained in our approach.
Additionally, chevrons can be nested and stacked on top of each other.
Stacked chevrons indicate parallel/overlapping execution of activities.
Nested chevrons relate groups of activities to each other.
In the given example, the first chevron indicates that C is executed in parallel to A, B, D, and E.
The two upper chevrons indicate that A and B are executed in parallel, but are executed before D and E, both of which are also executed in parallel. 
%By using chevrons, the approach allows to graphically represent interval orders in a familiar manner to the user.

The proposed approach assumes an interval order, representing a trace variant, as input and \emph{recursively} partitions the interval order by applying cuts to compute the layout of the visualization (cf. \autoref{fig:variant_visualization}).
In general, a cut is a partition of the nodes of a given interval order.
Based on the partition, induced interval orders are derived.
Each application of such a cut corresponds to chevrons and their positioning in the final visualization, e.g., stacked or side-by-side chevrons.
Nested chevrons result from the recursive manner.
Next, we define the computation of the proposed layout, i.e., we define two types of cuts.

\begin{figure}[tb]
    \centering
    \begin{subfigure}[t]{0.55\textwidth}
    \centering
    \begin{tikzpicture}[align=center,node distance=1cm] 
        \node[circle,fill=black,inner sep=0pt, minimum size=3pt,label=$v_1$](1){};
        \node[circle,fill=black,inner sep=0pt, minimum size=3pt,below=of 1,label=below:$v_2$](2){};
        \node[circle,fill=black,inner sep=0pt, minimum size=3pt,left=of 2,label=below:$v_0$](0){};
        \draw[->] (0) to (2);
        
        \node[circle,fill=black,inner sep=0pt, minimum size=3pt,right of=1,label=$v_3$](3){};
        \draw[->] (1) to (3);
        \draw[->] (2) to (3);
        \draw[->] (0) to (3);

        \node[circle,fill=black,inner sep=0pt, minimum size=3pt,right=of 3,label=$v_4$](4){};
        \node[circle,fill=black,inner sep=0pt, minimum size=3pt,below=of 4,label=below:$v_5$](5){};
        \draw[->] (3) to (4);
        \draw[->] (3) to (5);
        \draw[->] (0) to (4);
        \draw[->,bend right] (0) to (5);
        
        \draw[->,bend left=40] (1) to (4);
        \draw[->] (1) to (5);
        
        \draw[->] (2) to (4);
        \draw[->] (2) to (5);
        
        \draw[dashed,blue] (.5,-1.7) to (-1.5,-1.7) to (-1.5,1) to (.5,1) to  (.5,-1.7);
        \draw[dashed,blue] (.5,-1.7) to (1.5,-1.7) to (1.5,1) to (.5,1);
        \draw[dashed,blue] (1.5,-1.7) to (3,-1.7) to (3,1) to (1.5,1);
        
        \node[blue](v1) at (-.5,.7) {$V_1$};
        \node[blue](v1) at (1.,.7) {$V_2$};
        \node[blue](v1) at (2.25,.7) {$V_3$};
    \end{tikzpicture}
    \caption{Interval order and a maximal ordering cut}
    \end{subfigure}
    \hfill
    \begin{subfigure}[t]{0.44\textwidth}
    \centering
    \begin{tikzpicture}[align=center,node distance=1cm and 1cm] 
        \node[circle,fill=black,inner sep=0pt, minimum size=3pt,label=$v_1$](1){};
        \node[circle,fill=black,inner sep=0pt, minimum size=3pt,below=of 1,label=below:$v_2$](2){};
        \node[circle,fill=black,inner sep=0pt, minimum size=3pt,left=of 2,label=below:$v_0$](0){};
        \draw[->] (0) to (2);
    \end{tikzpicture}
    \caption{Induced interval order based on $V_1$}
    \end{subfigure}
    \caption{Example of an ordering cut, i.e., a partition of the nodes into $V_1{=}\{v_0,\allowbreak v_1,\allowbreak v_2\},\allowbreak V_2{=}\{v_3\},V_3{=}\{v_4,v_5\}$, and one corresponding induced interval order for $V_1$}
    \label{fig:example_ordering_cut}
\end{figure}
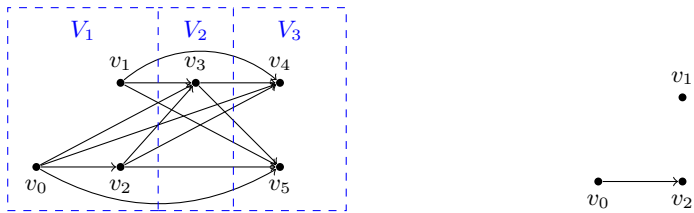

An \emph{ordering cut} partitions the activity instances into sets such that these sets can be totally ordered, i.e., all activity instances within a set can be related to all other activity instances from other sets.
In terms of the graph representation of an interval order, this implies that all nodes from one partition have a directed edge to all nodes from the other partition(s).
We depict an example of an ordering cut in \autoref{fig:example_ordering_cut}. 
Note that all nodes in $V_1$ are related to all nodes in $V_2$ and $V_3$.
%In general, $\forall v_1{\in}V_1,v_2{\in}V_2,v_3{\in}V_3 \big( v_1 {\leq} v_2 {\leq} v_3 \big)$.
Next, we formally define an ordering cut for an interval order.

% Reconsider the trace shown in \autoref{fig:example_instance_and_partial_order}.

\begin{definition}[Ordering Cut]
\label{def:ordering_cut}
Assume an interval order $(V,E,\lambda){\in}\mathcal{P}$. 
An ordering cut describes a partition of the nodes $V$ into $n {>} 1$ non-empty subsets $V_1,\dots,V_n$ such that: $\forall 1{\leq}i{<}j{\leq}n\Big( \forall v {\in} V_i,v' {\in} V_j \big((v,v'){\in}E\big) \Big)$.
\end{definition}

A \emph{parallel cut} indicates that activity instances from one partition overlap in time with activity instances in the other partition(s), i.e., activity instances from different partitions are unrelated to each other.
Thus, we are looking for \emph{components} in the graph representation of an interval order.

\begin{definition}[Parallel Cut]
\label{def:parallel_cut}
Assume an interval order $(V,E,\lambda){\in}\mathcal{P}$.
A parallel cut describes a partition of the nodes $V$ into $n {\geq} 1$ non-empty subsets $V_1,\dots,V_n$ such that $V_1,\dots,V_n$ represent connected components of $(V,E,\lambda)$, i.e., $\forall 1{\leq}i{<}j{\leq}n\allowbreak\Big(\forall v{\in}V_i \forall v' {\in} V_j\big( (v,v'){\notin}E \land (v',v){\notin}E   \big)\Big)$.
\end{definition}

We call a cut \emph{maximal} if $n$, i.e., the number of subsets, is maximal. 
%In the remainder of this paper, we always consider maximal cuts.

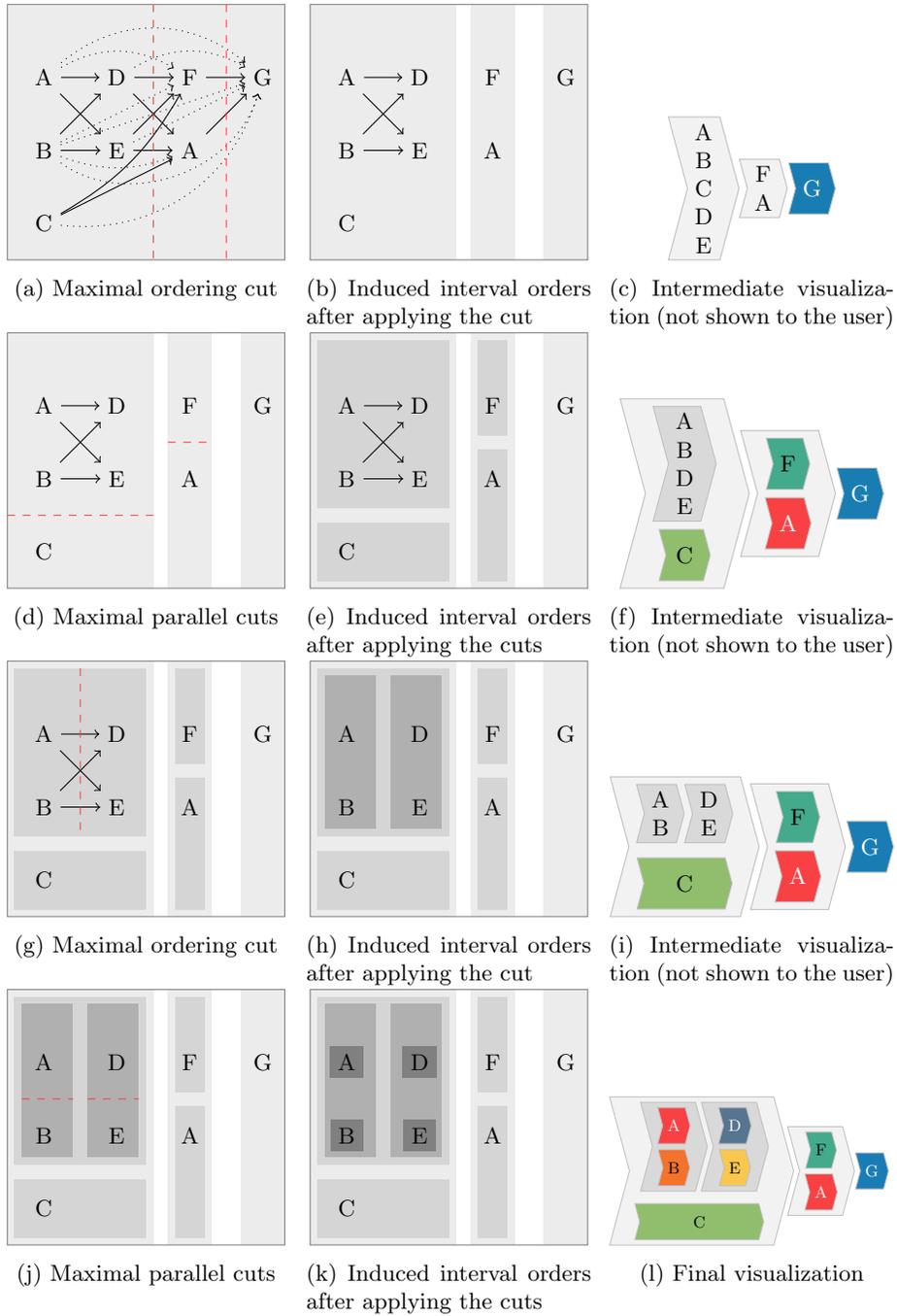
\begin{figure}
    \centering
    \begin{subfigure}[t]{0.32\textwidth}
        \centering
        \begin{tikzpicture}
            \draw[fill=black!15,draw=none,opacity=0.5] (-.5,1) rectangle (3.5 - 0.2,-2.5);
            \node[] (A) {A};
            \node[below of=A] (B) {B};
            \node[below of=B] (C) {C};
            \node[right of= A] (D) {D};
            \node[right of= B] (E) {E};
            
            \draw[->] (A) -- (D);
            \draw[->] (A) -- (E);
            \draw[->] (B) -- (D);
            \draw[->] (B) -- (E);
            
            \node[right of= D] (F) {F};
            \node[right of= E] (G) {A};        
            
            \draw[->] (D) -- (F);
            \draw[->] (E) -- (G);
            \draw[->] (D) -- (G);
            \draw[->] (E) -- (F);
            \draw[->, bend right=15] (C) to (F);
            \draw[->] (C) -- (G);
            \node[right of= F] (H) {G};     
            
            \draw[->,dotted] (B) -- (F);
            \draw[->,dotted] (B) -- (H);
            \draw[->] (F) -- (H);
            \draw[->] (G) -- (H);
            
            \draw[->, bend left, dotted] (A) to (F);
            \draw[->, bend left=40, dotted] (A) to (H);
            \draw[->, bend left, dotted] (D) to (H);
            
            \draw[->,dotted] (E) to (H);
            
            \draw[->, bend right=20, dotted] (B) to (G);
            \draw[->, bend right=55, dotted] (B) to (H);
            
            \draw[->, bend right=40, dotted] (C) to (H);
            
            \draw[dashed,red] (2.5,1) to (2.5,-2.5);
            \draw[dashed,red] (1.5,1) to (1.5,-2.5);

            \draw[gray] (-.5,1) to (3.5-.2,1) to (3.5-.2,-2.5) to (-.5,-2.5) to (-.5,1);

        \end{tikzpicture}
        \caption{Maximal ordering cut}
        \label{fig:example_cut_detection_a}
    \end{subfigure}
    \hfill
    \begin{subfigure}[t]{0.32\textwidth}
        \centering
        \begin{tikzpicture}
            % H
            \draw[fill=black!15,draw=none,opacity=0.5] (2.5 + 0.2,1) rectangle (3.5 - 0.2,-2.5);
            
            % F,G
            \draw[fill=black!15,draw=none,opacity=0.5] (1.5 + 0.2,1) rectangle (2.5 - 0.2,-2.5);
            
            % A,B,...,E
            \draw[fill=black!15,draw=none,opacity=0.5] (-.5 ,1) rectangle (1.5,-2.5);
        
            \node[] (A) {A};
            \node[below of=A] (B) {B};
            \node[below of=B] (C) {C};
            \node[right of= A] (D) {D};
            \node[right of= B] (E) {E};
            
            \draw[->] (A) -- (D);
            \draw[->] (A) -- (E);
            \draw[->] (B) -- (D);
            \draw[->] (B) -- (E);
            
            \node[right of= D] (F) {F};
            \node[right of= E] (G) {A};        
            
            %\draw[->] (D) -- (F);
            %\draw[->] (E) -- (G);
            %\draw[->] (D) -- (G);
            %\draw[->] (E) -- (F);
            %\draw[->, bend right=15] (C) to (F);
            %\draw[->] (C) -- (G);
            \node[right of= F] (H) {G};     
            
            %\draw[->] (F) -- (H);
            %\draw[->] (G) -- (H);
            
            %\draw[->, bend left, dotted] (A) to (F);
            %\draw[->, bend left=40, dotted] (A) to (H);
            %\draw[->, bend left, dotted] (D) to (H);
            
            %\draw[->,dotted] (E) to (H);
            
            %\draw[->, bend right=20, dotted] (B) to (G);
            %\draw[->, bend right=55, dotted] (B) to (H);
            
            %\draw[->, bend right=40, dotted] (C) to (H);
            
            %\draw[ultra thick] (1.5,1) to (1.5,-2.5);
            %\draw[ultra thick] (2.5,1) to (2.5,-2.5);
            \draw[gray] (-.5,1) to (3.5-.2,1) to (3.5-.2,-2.5) to (-.5,-2.5) to (-.5,1);        
            
        \end{tikzpicture}
        \caption{Induced interval orders after applying the cut}
        \label{fig:example_cut_detection_b}
    \end{subfigure}
    \hfill
    \begin{subfigure}[t]{0.32\textwidth}
        %\resizebox{\textwidth}{!}{%
        \centering
        \begin{tikzpicture}
          \tikzset{
            arrow/.style={
              draw=lightgray,
              minimum height=.7cm,
              inner sep=.3em,
              shape=signal,
              signal from=west,
              signal to=east,
              signal pointer angle=150,
            }
          }
          %\node[arrow] {first};
          \node[arrow,align=center,fill=black!5] (1) 
            {A\\B\\C\\D\\E};
          
          \node[arrow,align=center,right=.1cm of 1,fill=black!5] (2) {F\\A};
          
          \node[arrow,align=center,right=.1cm of 2,fill=blue2,text=white] (3) {G};
        \end{tikzpicture}
        %}
        \caption{Intermediate visualization (not shown to the user)}
        \label{fig:example_cut_detection_c}
    \end{subfigure}
    \hfill
    \begin{subfigure}[t]{0.32\textwidth}
        \centering
        \begin{tikzpicture}
            \draw[fill=black!15,draw=none,opacity=0.5] (2.5 + 0.2,1) rectangle (3.5 - 0.2,-2.5);
            \draw[fill=black!15,draw=none,opacity=0.5] (1.5 + 0.2,1) rectangle (2.5 - 0.2,-2.5);
            \draw[fill=black!15,draw=none,opacity=0.5] (-.5 ,1) rectangle (1.5,-2.5);
            
            \node[] (A) {A};
            \node[below of=A] (B) {B};
            \node[below of=B] (C) {C};
            \node[right of= A] (D) {D};
            \node[right of= B] (E) {E};
            
            \draw[->] (A) -- (D);
            \draw[->] (A) -- (E);
            \draw[->] (B) -- (D);
            \draw[->] (B) -- (E);
            
            \node[right of= D] (F) {F};
            \node[right of= E] (G) {A};        
            
            %\draw[->] (D) -- (F);
            %\draw[->] (E) -- (G);
            %\draw[->] (D) -- (G);
            %\draw[->] (E) -- (F);
            %\draw[->, bend right=15] (C) to (F);
            %\draw[->] (C) -- (G);
            \node[right of= F] (H) {G};     
            
            %\draw[->] (F) -- (H);
            %\draw[->] (G) -- (H);
            
            %\draw[->, bend left, dotted] (A) to (F);
            %\draw[->, bend left=40, dotted] (A) to (H);
            %\draw[->, bend left, dotted] (D) to (H);
            
            %\draw[->,dotted] (E) to (H);
            
            %\draw[->, bend right=20, dotted] (B) to (G);
            %\draw[->, bend right=55, dotted] (B) to (H);
            
            %\draw[->, bend right=40, dotted] (C) to (H);
            
            %\draw[ultra thick] (1.5,1) to (1.5,-2.5);
            %\draw[ultra thick] (2.5,1) to (2.5,-2.5);
            \draw[gray] (-.5,1) to (3.5-.2,1) to (3.5-.2,-2.5) to (-.5,-2.5) to (-.5,1);    \draw[dashed,red] (-0.5,-1.5) to (1.7-.2,-1.5);
            \draw[dashed,red] (1.5+.2,-0.5) to (2.5-.2,-0.5);
        \end{tikzpicture}
        \caption{Maximal parallel cuts}
        \label{fig:example_cut_detection_d}
    \end{subfigure}
    \hfill
    \begin{subfigure}[t]{0.32\textwidth}
        \centering
        \begin{tikzpicture}
            \draw[fill=black!15,draw=none,opacity=0.5] (2.5 + 0.2,1) rectangle (3.5 - 0.2,-2.5);
            
            % F,G
            \draw[fill=black!15,draw=none,opacity=0.5] (1.5 + 0.2,1) rectangle (2.5 - 0.2,-2.5);
            \draw[fill=black!25,draw=none,opacity=0.5] (1.5+.1 + 0.2,1-.1) rectangle (2.5-.1 - 0.2,-.5+.1);
            
            \draw[fill=black!25,draw=none,opacity=0.5] (1.5+.1 + 0.2, -.5-.1) rectangle (2.5-.1 - 0.2,-2.5+.1);
            
            %A,B,C,D,E
            \draw[fill=black!15,draw=none,opacity=0.5] (-.5 ,1) rectangle (1.5,-2.5);
            %A,D,B,E
            \draw[fill=black!25,draw=none,opacity=0.5] (-.5+.1 ,1-.1) rectangle (1.5 - 0.1,-1.5+.1);
            %C
            \draw[fill=black!25,draw=none,opacity=0.5] (-.5+.1 ,-1.5-.1) rectangle (1.5 -.1,-2.5+.1);

            \node[] (A) {A};
            \node[below of=A] (B) {B};
            \node[below of=B] (C) {C};
            \node[right of= A] (D) {D};
            \node[right of= B] (E) {E};
            
            \draw[->] (A) -- (D);
            \draw[->] (A) -- (E);
            \draw[->] (B) -- (D);
            \draw[->] (B) -- (E);
            
            \node[right of= D] (F) {F};
            \node[right of= E] (G) {A};        
            
            %\draw[->] (D) -- (F);
            %\draw[->] (E) -- (G);
            %\draw[->] (D) -- (G);
            %\draw[->] (E) -- (F);
            %\draw[->, bend right=15] (C) to (F);
            %\draw[->] (C) -- (G);
            \node[right of= F] (H) {G};     
            
            %\draw[->] (F) -- (H);
            %\draw[->] (G) -- (H);
            
            %\draw[->, bend left, dotted] (A) to (F);
            %\draw[->, bend left=40, dotted] (A) to (H);
            %\draw[->, bend left, dotted] (D) to (H);
            
            %\draw[->,dotted] (E) to (H);
            
            %\draw[->, bend right=20, dotted] (B) to (G);
            %\draw[->, bend right=55, dotted] (B) to (H);
            
            %\draw[->, bend right=40, dotted] (C) to (H);
            
            %\draw[ultra thick] (1.5,1) to (1.5,-2.5);
            %\draw[ultra thick] (2.5,1) to (2.5,-2.5);
            \draw[gray] (-.5,1) to (3.5-.2,1) to (3.5-.2,-2.5) to (-.5,-2.5) to (-.5,1);
            %\draw[thick] (-0.3,-1.5) to (1.5,-1.5);
            %\draw[thick] (1.5,-0.5) to (2.5,-0.5);
        \end{tikzpicture}
        \caption{Induced interval orders after applying the cuts}
        \label{fig:example_cut_detection_e}
    \end{subfigure}
    \hfill
    \begin{subfigure}[t]{0.32\textwidth}
        %\resizebox{\textwidth}{!}{%
        \centering
        \begin{tikzpicture}
          \tikzset{
            arrow/.style={
              draw=lightgray,
              minimum height=.7cm,
              inner sep=.3em,
              shape=signal,
              signal from=west,
              signal to=east,
              signal pointer angle=150,
            }
          }
          %\node[arrow] {first};
          \node[arrow,align=center,fill=black!5] (1) 
            {\begin{tikzpicture}
                \node[arrow, fill=black!15] (AB) {A\\B\\D\\E};
                \node[arrow,minimum width=.7cm,xshift=-0cm,below=.1cm of AB,fill=green1] (C) {C};
            \end{tikzpicture}};
          
          \node[arrow,align=center,right=.1cm of 1,fill=black!5] (2)   
            {\begin{tikzpicture}
                \node[arrow,fill=green2] (F) {F};
                \node[arrow, below=.1cm of F,fill=red,text=white] (G) {A};
            \end{tikzpicture}};
          
          \node[arrow,align=center,right=.1cm of 2,fill=blue2,text=white] (3) {G};
        \end{tikzpicture}
        %}
        \caption{Intermediate visualization (not shown to the user)}
        \label{fig:example_cut_detection_f}
    \end{subfigure}
    \hfill
    \begin{subfigure}[t]{0.32\textwidth}
        \centering
        \begin{tikzpicture}
            
            % H
            \draw[fill=black!15,draw=none,opacity=0.5] (2.5 + 0.2,1) rectangle (3.5 - 0.2,-2.5);
            
            % F,G
            \draw[fill=black!15,draw=none,opacity=0.5] (1.5 + 0.2,1) rectangle (2.5 - 0.2,-2.5);
            %F
            \draw[fill=black!25,draw=none,opacity=0.5] (1.5+.1 + 0.2,1-.1) rectangle (2.5-.1 - 0.2,-.5+.1);
            %G
            \draw[fill=black!25,draw=none,opacity=0.5] (1.5+.1 + 0.2, -.5-.1) rectangle (2.5-.1 - 0.2,-2.5+.1);
            %A,B,C,D,E
            \draw[fill=black!15,draw=none,opacity=0.5] (-.5 ,1) rectangle (1.5,-2.5);
            %A,D,B,E
            \draw[fill=black!25,draw=none,opacity=0.5] (-.5+.1 ,1-.1) rectangle (1.5 - 0.1,-1.5+.1);
            %C
            \draw[fill=black!25,draw=none,opacity=0.5] (-.5+.1 ,-1.5-.1) rectangle (1.5 -.1,-2.5+.1);
            
            \node[] (A) {A};
            \node[below of=A] (B) {B};
            \node[below of=B] (C) {C};
            \node[right of= A] (D) {D};
            \node[right of= B] (E) {E};
            
            \draw[->] (A) -- (D);
            \draw[->] (A) -- (E);
            \draw[->] (B) -- (D);
            \draw[->] (B) -- (E);
            
            \node[right of= D] (F) {F};
            \node[right of= E] (G) {A};        
            
            %\draw[->] (D) -- (F);
            %\draw[->] (E) -- (G);
            %\draw[->] (D) -- (G);
            %\draw[->] (E) -- (F);
            %\draw[->, bend right=15] (C) to (F);
            %\draw[->] (C) -- (G);
            \node[right of= F] (H) {G};     
            
            %\draw[->] (F) -- (H);
            %\draw[->] (G) -- (H);
            
            %\draw[->, bend left, dotted] (A) to (F);
            %\draw[->, bend left=40, dotted] (A) to (H);
            %\draw[->, bend left, dotted] (D) to (H);
            
            %\draw[->,dotted] (E) to (H);
            
            %\draw[->, bend right=20, dotted] (B) to (G);
            %\draw[->, bend right=55, dotted] (B) to (H);
            
            %\draw[->, bend right=40, dotted] (C) to (H);
            
            %\draw[ultra thick] (1.5,1) to (1.5,-2.5);
            %\draw[ultra thick] (2.5,1) to (2.5,-2.5);
            \draw[gray] (-.5,1) to (3.5-.2,1) to (3.5-.2,-2.5) to (-.5,-2.5) to (-.5,1);
            %\draw[thick] (-0.3,-1.5) to (1.5,-1.5);
            \draw[dashed,red] (.5,1-.1) to (.5,-1.5+.1);
            %\draw[thick] (1.5,-0.5) to (2.5,-0.5);
        \end{tikzpicture}
        \caption{Maximal ordering cut}
        \label{fig:example_cut_detection_g}
    \end{subfigure}
    \hfill
    \begin{subfigure}[t]{0.32\textwidth}
        \centering
        \begin{tikzpicture}
            % H
            \draw[fill=black!15,draw=none,opacity=0.5] (2.5 + 0.2,1) rectangle (3.5 - 0.2,-2.5);
            
            % F,G
            \draw[fill=black!15,draw=none,opacity=0.5] (1.5 + 0.2,1) rectangle (2.5 - 0.2,-2.5);
            %F
            \draw[fill=black!25,draw=none,opacity=0.5] (1.5+.1 + 0.2,1-.1) rectangle (2.5-.1 - 0.2,-.5+.1);
            %G
            \draw[fill=black!25,draw=none,opacity=0.5] (1.5+.1 + 0.2, -.5-.1) rectangle (2.5-.1 - 0.2,-2.5+.1);
            %A,B,C,D,E
            \draw[fill=black!15,draw=none,opacity=0.5] (-.5 ,1) rectangle (1.5,-2.5);
            %A,D,B,E
            \draw[fill=black!25,draw=none,opacity=0.5] (-.5+.1 ,1-.1) rectangle (1.5 - 0.1,-1.5+.1);
            %C
            \draw[fill=black!25,draw=none,opacity=0.5] (-.5+.1 ,-1.5-.1) rectangle (1.5 -.1,-2.5+.1);
            % A,B
            \draw[fill=black!45,draw=none,opacity=0.5] (-.5+.1+.1 ,1-.1-.1) rectangle (.7-.1 - 0.2,-1.5+.1+.1);
            % D,E
            \draw[fill=black!45,draw=none,opacity=0.5] (0.4+.1+.1 ,1-.1-.1) rectangle (1.5-.2 ,-1.5+.1+.1);
            
            \node[] (A) {A};
            \node[below of=A] (B) {B};
            \node[below of=B] (C) {C};
            \node[right of= A] (D) {D};
            \node[right of= B] (E) {E};
            
            %\draw[->] (A) -- (D);
            %\draw[->] (A) -- (E);
            %\draw[->] (B) -- (D);
            %\draw[->] (B) -- (E);
            
            \node[right of= D] (F) {F};
            \node[right of= E] (G) {A};        
            
            %\draw[->] (D) -- (F);
            %\draw[->] (E) -- (G);
            %\draw[->] (D) -- (G);
            %\draw[->] (E) -- (F);
            %\draw[->, bend right=15] (C) to (F);
            %\draw[->] (C) -- (G);
            \node[right of= F] (H) {G};     
            
            %\draw[->] (F) -- (H);
            %\draw[->] (G) -- (H);
            
            %\draw[->, bend left, dotted] (A) to (F);
            %\draw[->, bend left=40, dotted] (A) to (H);
            %\draw[->, bend left, dotted] (D) to (H);
            
            %\draw[->,dotted] (E) to (H);
            
            %\draw[->, bend right=20, dotted] (B) to (G);
            %\draw[->, bend right=55, dotted] (B) to (H);
            
            %\draw[->, bend right=40, dotted] (C) to (H);
            
            %\draw[ultra thick] (1.5,1) to (1.5,-2.5);
            %\draw[ultra thick] (2.5,1) to (2.5,-2.5);
            \draw[gray] (-.5,1) to (3.5-.2,1) to (3.5-.2,-2.5) to (-.5,-2.5) to (-.5,1);
            %\draw[thick] (-0.3,-1.5) to (1.5,-1.5);
            %\draw[gray] (.5,1) to (.5,-1.5);
            %\draw[thick] (1.5,-0.5) to (2.5,-0.5);
        \end{tikzpicture}
        \caption{Induced interval orders after applying the cut}
        \label{fig:example_cut_detection_h}
    \end{subfigure}
    \hfill
    \begin{subfigure}[t]{0.32\textwidth}
        %\resizebox{\textwidth}{!}{%
        \centering
        \begin{tikzpicture}
          \tikzset{
            arrow/.style={
              draw=lightgray,
              minimum height=.7cm,
              inner sep=.3em,
              shape=signal,
              signal from=west,
              signal to=east,
              signal pointer angle=150,
            }
          }
          %\node[arrow] {first};
          \node[arrow,align=center,fill=black!5] (1) 
            {\begin{tikzpicture}
                \node[arrow, fill=black!15] (AB) 
                    {A\\B};
                \node[arrow, right=.1cm of AB,fill=black!15] (DE) 
                    {D\\E};
                \coordinate (CENTER) at ($(AB)!0.5!(DE)$);
                \node[arrow,minimum width=1.3cm,below=.6cm of CENTER,fill=green1] (C) {C};
            \end{tikzpicture}};
          
          \node[arrow,align=center,right=.1cm of 1,fill=black!5] (2)   
            {\begin{tikzpicture}
                \node[arrow,fill=green2] (F) {F};
                \node[arrow, below=.1cm of F,fill=red,text=white] (G) {A};
            \end{tikzpicture}};
          
          \node[arrow,align=center,right=.1cm of 2,fill=blue2,text=white] (3) {G};
        \end{tikzpicture}
        %}
        \caption{Intermediate visualization (not shown to the user)}
        \label{fig:example_cut_detection_i}
    \end{subfigure}
    \hfill
    \begin{subfigure}[t]{0.32\textwidth}
        \centering
        \begin{tikzpicture}
            % H
            \draw[fill=black!15,draw=none,opacity=0.5] (2.5 + 0.2,1) rectangle (3.5 - 0.2,-2.5);
            
            % F,G
            \draw[fill=black!15,draw=none,opacity=0.5] (1.5 + 0.2,1) rectangle (2.5 - 0.2,-2.5);
            %F
            \draw[fill=black!25,draw=none,opacity=0.5] (1.5+.1 + 0.2,1-.1) rectangle (2.5-.1 - 0.2,-.5+.1);
            %G
            \draw[fill=black!25,draw=none,opacity=0.5] (1.5+.1 + 0.2, -.5-.1) rectangle (2.5-.1 - 0.2,-2.5+.1);
            %A,B,C,D,E
            \draw[fill=black!15,draw=none,opacity=0.5] (-.5 ,1) rectangle (1.5,-2.5);
            %A,D,B,E
            \draw[fill=black!25,draw=none,opacity=0.5] (-.5+.1 ,1-.1) rectangle (1.5 - 0.1,-1.5+.1);
            %C
            \draw[fill=black!25,draw=none,opacity=0.5] (-.5+.1 ,-1.5-.1) rectangle (1.5 -.1,-2.5+.1);
            % A,B
            \draw[fill=black!45,draw=none,opacity=0.5] (-.5+.1+.1 ,1-.1-.1) rectangle (.7-.1 - 0.2,-1.5+.1+.1);
            % D,E
            \draw[fill=black!45,draw=none,opacity=0.5] (0.4+.1+.1 ,1-.1-.1) rectangle (1.5-.2 ,-1.5+.1+.1);
        
            \node[] (A) {A};
            \node[below of=A] (B) {B};
            \node[below of=B] (C) {C};
            \node[right of= A] (D) {D};
            \node[right of= B] (E) {E};
            
            %\draw[->] (A) -- (D);
            %\draw[->] (A) -- (E);
            %\draw[->] (B) -- (D);
            %\draw[->] (B) -- (E);
            
            \node[right of= D] (F) {F};
            \node[right of= E] (G) {A};        
            
            %\draw[->] (D) -- (F);
            %\draw[->] (E) -- (G);
            %\draw[->] (D) -- (G);
            %\draw[->] (E) -- (F);
            %\draw[->, bend right=15] (C) to (F);
            %\draw[->] (C) -- (G);
            \node[right of= F] (H) {G};     
            
            %\draw[->] (F) -- (H);
            %\draw[->] (G) -- (H);
            
            %\draw[->, bend left, dotted] (A) to (F);
            %\draw[->, bend left=40, dotted] (A) to (H);
            %\draw[->, bend left, dotted] (D) to (H);
            
            %\draw[->,dotted] (E) to (H);
            
            %\draw[->, bend right=20, dotted] (B) to (G);
            %\draw[->, bend right=55, dotted] (B) to (H);
            
            %\draw[->, bend right=40, dotted] (C) to (H);
            
            %\draw[ultra thick] (1.5,1) to (1.5,-2.5);
            %\draw[ultra thick] (2.5,1) to (2.5,-2.5);
            \draw[gray] (-.5,1) to (3.5-.2,1) to (3.5-.2,-2.5) to (-.5,-2.5) to (-.5,1);
            %\draw[thick] (-0.3,-1.5) to (1.5,-1.5);
            %\draw[] (.5,1) to (.5,-1.5);
            %\draw[thick] (1.5,-0.5) to (2.5,-0.5);
            \draw[red, dashed] (-.5+.1+.1,-0.5) to (.7- 0.3,-0.5);
            \draw[red, dashed] (0.5+.1,-0.5) to (1.5-.2 ,-0.5);
        \end{tikzpicture}
        \caption{Maximal parallel cuts}
        \label{fig:example_cut_detection_j}
    \end{subfigure}
    \hfill
    \begin{subfigure}[t]{0.32\textwidth}
        \centering
        \begin{tikzpicture}
            % H
            \draw[fill=black!15,draw=none,opacity=0.5] (2.5 + 0.2,1) rectangle (3.5 - 0.2,-2.5);
            
            % F,G
            \draw[fill=black!15,draw=none,opacity=0.5] (1.5 + 0.2,1) rectangle (2.5 - 0.2,-2.5);
            %F
            \draw[fill=black!25,draw=none,opacity=0.5] (1.5+.1 + 0.2,1-.1) rectangle (2.5-.1 - 0.2,-.5+.1);
            %G
            \draw[fill=black!25,draw=none,opacity=0.5] (1.5+.1 + 0.2, -.5-.1) rectangle (2.5-.1 - 0.2,-2.5+.1);
            %A,B,C,D,E
            \draw[fill=black!15,draw=none,opacity=0.5] (-.5 ,1) rectangle (1.5,-2.5);
            %A,D,B,E
            \draw[fill=black!25,draw=none,opacity=0.5] (-.5+.1 ,1-.1) rectangle (1.5 - 0.1,-1.5+.1);
            %C
            \draw[fill=black!25,draw=none,opacity=0.5] (-.5+.1 ,-1.5-.1) rectangle (1.5 -.1,-2.5+.1);
            % A,B
            \draw[fill=black!45,draw=none,opacity=0.5] (-.5+.1+.1 ,1-.1-.1) rectangle (.7-.1 - 0.2,-1.5+.1+.1);
            % D,E
            \draw[fill=black!45,draw=none,opacity=0.5] (0.4+.1+.1 ,1-.1-.1) rectangle (1.5-.2 ,-1.5+.1+.1);
            
            \node[fill=black!50,draw=none] (A) {A};
            \node[below of=A,fill=black!50,draw=none] (B) {B};
            \node[below of=B] (C) {C};
            \node[right of= A,fill=black!50,draw=none] (D) {D};
            \node[right of= B,fill=black!50,draw=none] (E) {E};
            
            %\draw[->] (A) -- (D);
            %\draw[->] (A) -- (E);
            %\draw[->] (B) -- (D);
            %\draw[->] (B) -- (E);
            
            \node[right of= D] (F) {F};
            \node[right of= E] (G) {A};        
            
            %\draw[->] (D) -- (F);
            %\draw[->] (E) -- (G);
            %\draw[->] (D) -- (G);
            %\draw[->] (E) -- (F);
            %\draw[->, bend right=15] (C) to (F);
            %\draw[->] (C) -- (G);
            \node[right of= F] (H) {G};     
            
            %\draw[->] (F) -- (H);
            %\draw[->] (G) -- (H);
            
            %\draw[->, bend left, dotted] (A) to (F);
            %\draw[->, bend left=40, dotted] (A) to (H);
            %\draw[->, bend left, dotted] (D) to (H);
            
            %\draw[->,dotted] (E) to (H);
            
            %\draw[->, bend right=20, dotted] (B) to (G);
            %\draw[->, bend right=55, dotted] (B) to (H);
            
            %\draw[->, bend right=40, dotted] (C) to (H);
            
            %\draw[ultra thick] (1.5,1) to (1.5,-2.5);
            %\draw[ultra thick] (2.5,1) to (2.5,-2.5);
            \draw[gray] (-.5,1) to (3.5-.2,1) to (3.5-.2,-2.5) to (-.5,-2.5) to (-.5,1);
            %\draw[thick] (-0.3,-1.5) to (1.5,-1.5);
            %\draw[] (.5,1) to (.5,-1.5);
            %\draw[thick] (1.5,-0.5) to (2.5,-0.5);
            %\draw[gray,very thin] (-.5,-0.5) to (1.5,-0.5);
        \end{tikzpicture}
        \caption{Induced interval orders after applying the cuts}
        \label{fig:example_cut_detection_k}
    \end{subfigure}
    \hfill
    \begin{subfigure}[t]{0.32\textwidth}
        \resizebox{\textwidth}{!}{%
        \centering
        \begin{tikzpicture}
          \tikzset{
            arrow/.style={
              draw=lightgray,
              minimum height=.7cm,
              inner sep=.3em,
              shape=signal,
              signal from=west,
              signal to=east,
              signal pointer angle=150,
            }
          }
          %\node[arrow] {first};
          \node[arrow,align=center,fill=black!5] (1) 
            {\begin{tikzpicture}
                \node[arrow, fill=black!15] (AB) 
                    {\begin{tikzpicture}
                        \node[arrow,fill=red,text=white] (A) {A};
                        \node[arrow,below=.1cm of A,fill=orange1] (B) {B};
                    \end{tikzpicture}};
                \node[arrow, right=.1cm of AB,fill=black!15] (DE) 
                    {\begin{tikzpicture}
                        \node[arrow,fill=blue1,text=white] (D) {D};
                        \node[arrow,below=.1cm of D,fill=yellow] (E) {E};
                    \end{tikzpicture}};
                \coordinate (CENTER) at ($(AB)!0.5!(DE)$);
                \node[arrow,minimum width=2.5cm,xshift=-.1cm,below=1.1cm of CENTER,fill=green1] (C) {C};
            \end{tikzpicture}};
          
          \node[arrow,align=center,right=.1cm of 1,fill=black!5] (2)   
            {\begin{tikzpicture}
                \node[arrow,fill=green2] (F) {F};
                \node[arrow, below=.1cm of F,fill=red,text=white] (G) {A};
            \end{tikzpicture}};
          
          \node[arrow,align=center,right=.1cm of 2,fill=blue2,text=white] (3) {G};
        \end{tikzpicture}
        }
        \caption{Final visualization}
        \label{fig:example_cut_detection_l}
    \end{subfigure}
    \caption{Example of recursively applying ordering and parallel cuts to an interval order and the corresponding visualization}
    \label{fig:example_cut_detection}
\end{figure}

\autoref{fig:example_cut_detection} shows an example of the proposed visualization approach. 
We use the interval order from \autoref{fig:partial-order} as input.
The visualization approach recursively looks for a maximal ordering or parallel cut.
In the example, we initially find an ordering cut of size three (cf. \autoref{fig:example_cut_detection_a}).
Given the cut, we create three induced interval orders (cf. \autoref{fig:example_cut_detection_b}).
As stated before, each induced interval order created by a cut represents a chevron.
In general, an ordering cut indicates the horizontal alignment of chevrons while a parallel cut indicates the vertical alignment of chevrons.
Since we found an ordering cut of size three, the intermediate visualization consists of three horizontally-aligned chevrons (cf. \autoref{fig:example_cut_detection_c}).
If an induced interval order only consists of one element (e.g., the third induced interval order in \autoref{fig:example_cut_detection_b}), we fill the corresponding chevron with a color that is unique for the given activity label (cf. \autoref{fig:example_cut_detection_c}).
As in the classic trace variant explorer, colors are used to better distinguish different activity labels.

We now recursively apply cuts to the induced interval orders. % until no more cuts can be applied.
In the first two interval orders, we apply a parallel cut (cf. \autoref{fig:example_cut_detection_d}). 
The third interval order consists only of one node labeled with $G$; thus, no further cuts can be applied.
\autoref{fig:example_cut_detection_e} shows the induced interval orders after applying the two parallel cuts.
As stated before, time-overlapping activity instances are indicated by stacked chevrons. 
Since both applied parallel cuts have size two, we create two stacked chevrons each within the first and the second chevron (cf. \autoref{fig:example_cut_detection_f}). 
After another ordering cut (cf. \autoref{fig:example_cut_detection_g}-\ref{fig:example_cut_detection_i}) and two more parallel cuts (cf. \autoref{fig:example_cut_detection_j}), the visualization approach stops because all induced interval orders consist of only one activity instance.
\autoref{fig:example_cut_detection_l} shows the final visualization.

\subsection{Formal Guarantees}
\label{sec:guarantees}
Next, we show that the proposed approach is deterministic, i.e., the same visualization is always returned for the same interval order.
We therefore show that different cuts cannot coexist, i.e., either a parallel cut, an ordering cut, or no cut exists in an interval order.
Further, we show that maximal cuts are unique.

\begin{lemma}[Cuts Cannot Coexist]
\label{lem:coexistence}
In an interval order $(V,E,\lambda){\in}\mathcal{P}$ a parallel and an ordering cut cannot coexist.
\end{lemma}

\begin{proof}
Let $(V,E,\lambda) {\in} \mathcal{P}$ be an interval order with an ordering cut $V_1,\dots,V_n$ for some $n{\geq}2$.
Assume there exists a parallel cut, too, i.e., $V_1',\dots,V_m'$ for some $m{\geq}2$.
For $1{\leq}j{\leq}m$, assume that for an arbitrary $v{\in}V$ it holds that $v{\in}V_j'$ such that $v{\in}V_i$ for some $i{\in}\{1,\dots,n\}$.
Since an ordering cut exists, we know that $\forall w{\in}V_{i+1}{\cup}\dots{\cup}V_n\big((v,w){\in}E\big)$ and $\forall w'{\in}V_{1}{\cup}\dots{\cup}V_{i-1}\big(\allowbreak(w',v){\in} E\big)$.
Since $V_1',\dots,V_m'$ is a parallel cut, i.e., each $V_k'{\in}\{V_1',\dots,V_m'\}$ represents a connected component (\autoref{def:parallel_cut}), also all $w$ and $w'$ must be in $V_j'$.
Hence, $V_j'{=}\{v\}{\cup}V_{1}{\cup}\dots{\cup}V_{i-1}{\cup}\allowbreak V_{i+1}{\cup}\allowbreak\dots{\cup}V_n$.
Further, since $\forall w'{\in}V_{1}{\cup}\dots{\cup}V_{i-1}\forall w{\in}V_{i+1}{\cup}\dots{\cup}V_n\big((w',w){\in}E \land (w',v){\in}E \land (v,w){\in}E\big)$ it follows by \autoref{def:parallel_cut} that $V_j'{=}V_1{\cup}\dots{\cup}V_n{=}V$. Hence, $\forall V_k'{\in}\{V_1',\dots,V_m'\} {\setminus} \{V_j'\} \allowbreak (V_k'{=}\emptyset)$ since $V_1',\dots,V_m'$ is a partition of $V$. 
This contradicts our assumption that there exists a parallel cut, too. 
The other direction is symmetrical.$\qed$
\end{proof}

Since cuts cannot coexist (cf. \autoref{lem:coexistence}), one cut is applicable for a given interval order at most.
Next, we show that maximal cuts are unique. 

\begin{lemma}[Maximal Ordering Cuts Are Unique]
\label{lem:ordering_unique}
If an ordering cut exists in a given interval order $(V,E,\lambda){\in}\mathcal{P}$, the maximal ordering cut is unique. %, i.e., there is a unique partition of $V$.
\end{lemma}

\begin{proof}
Proof by contradiction. Assume an interval order $(V,E,\lambda) {\in} \mathcal{P}$ having two \emph{different} maximal ordering cuts, i.e., $V_1,\dots,V_n$ and $V'_1,\dots,V'_n$.

\noindent $\Rightarrow \exists i {\in} \{1,\dots,n\} \forall j {\in} \{1,\dots,n\} \big(  V_i {\neq} V_j' \big) \Rightarrow V_i{\neq} V_i'$

\noindent $\Rightarrow \exists v {\in} V_i {\cup} V_i' \big( (v{\in}V_i \land v{\notin}V'_i) \lor (v{\notin}V_i \land v{\in}V'_i) \big)$

\noindent Assume $v {\in} V_i \land v {\notin} V_i'$ (the other case is symmetric)

\noindent $\Rightarrow v{\in} V_1' {\cup} \dots {\cup} V_{i-1}' {\cup} V_{i+1}' {\cup} \dots {\cup} V_n' {=} V{\setminus}V_i'$

\noindent \begin{minipage}[t]{0.49\textwidth}
\noindent 1) Assume $v {\in}  V_{1}' {\cup} \dots {\cup} V_{i-1}'$

\noindent $\xRightarrow[]{\autoref{def:ordering_cut}} \forall v' {\in} V_i \big( (v,v') {\in} E \big)$

\end{minipage}
\hfill
\begin{minipage}[t]{0.49\textwidth}
\noindent 2) Assume $v {\in}  V_{i+1}' {\cup} \dots {\cup} V_n'$

\noindent $\xRightarrow[]{\autoref{def:ordering_cut}} \forall v' {\in} V_i \big( (v',v) {\in} E \big)$

\end{minipage}

\noindent $\xRightarrow[]{v{\in}V_i} (v,v) {\in} E$ contradicts the assumption $(V,E,\lambda)$ represents an interval order because irreflexivity is not satisfied. $\qed$

\end{proof}

\begin{lemma}[Maximal Parallel Cuts Are Unique]
\label{lem:parallel_unique}
If a parallel cut exists in a given interval order $(V,E,\lambda){\in}\mathcal{P}$, the maximal parallel cut is unique. %, i.e., there is a unique partition of $V$.
\end{lemma}

\begin{proof}[\autoref{lem:parallel_unique}]
By definition, components of a graph are unique. $\qed$
\end{proof}

\autoref{lem:ordering_unique} and \autoref{lem:parallel_unique} show that maximal cuts, both ordering and parallel, are unique. 
Together with \autoref{lem:coexistence}, we derive that the proposed visualization approach is \emph{deterministic}, i.e., the approach always returns the same visualization for the same input, because for a given interval order only one cut type is applicable at most and if a cut exists, the maximal cut is unique.

\subsection{Limitations}
\label{sec:limitations}

In this section, we discuss the limitations of the proposed visualization approach. 

Reconsider the example in \autoref{fig:example_cut_detection}.
Cuts are recursively applied until one node, i.e., an activity instance, remains in each induced interval order (cf. \autoref{fig:example_cut_detection_k}).
However, there are certain cases in which the proposed approach cannot apply cuts although more than one node exists in an (induced) interval order.
%In these cases, we are not able to visualize all

\vspace{-.4cm}

\begin{figure}[]
    \centering
    \begin{subfigure}[b]{0.54\textwidth}
    \centering
    \resizebox{\textwidth}{!}{%
    \begin{tikzpicture}
        \tikzstyle{labels}=[font=\footnotesize,color=gray]
        %\node[align=center](text) at (4.5,3.1) {\emph{Case/Process Instance 2}};
        %\node[align=center](text2) at (4.5,3.5) {};
        
        \draw[dashed,lightgray] (0,.5) to (0,3);
        %\node[labels] at (0,.3) {08:00};
        \draw[dashed,lightgray] (1,.5) to (1,3);
        %\node[labels] at (1,.3) {09:00};
        \draw[dashed,lightgray] (2,.5) to (2,3);
        %\node[labels] at (2,.3) {10:00};
        \draw[dashed,lightgray] (3,.5) to (3,3);
        %\node[labels] at (3,.3) {11:00};
        \draw[dashed,lightgray] (4,.5) to (4,3);
        %\node[labels] at (4,.3) {12:00};
        \draw[dashed,lightgray] (5,.5) to (5,3);
        %\node[labels] at (5,.3) {13:00};
        \draw[dashed,lightgray] (6,.5) to (6,3);
        %\node[labels] at (6,.3) {14:00};
        \draw[dashed,lightgray] (7,.5) to (7,3);
        %\node[labels] at (7,.3) {15:00};
        \draw[dashed,lightgray] (8,.5) to (8,3);
        %\node[labels] at (8,.3) {16:00};
        %\draw[dashed,lightgray] (9,.5) to (9,3);
        %\node[labels] at (9,.3) {17:00};
        
        \node[circle,fill=black,inner sep=0pt, minimum size=3pt] (A+S) at (0,1) {};
        \node[circle,fill=black,inner sep=0pt, minimum size=3pt] (A+C) at (1.5,1) {};
        \draw[-] (A+S) to node[above] {A} (A+C);
        
        \node[circle,fill=black,inner sep=0pt, minimum size=3pt] (B+S) at (1,1.5) {};
        \node[circle,fill=black,inner sep=0pt, minimum size=3pt] (B+C) at (2.5,1.5) {};
        \draw[-] (B+S) to node[above] {B} (B+C);
        
        \node[circle,fill=black,inner sep=0pt, minimum size=3pt] (C+S) at (2,1) {};
        \node[circle,fill=black,inner sep=0pt, minimum size=3pt] (C+C) at (3.5,1) {};
        \draw[-] (C+S) to node[above] {C} (C+C);
        
        \node[circle,fill=black,inner sep=0pt, minimum size=3pt] (D+S) at (3,1.5) {};
        \node[circle,fill=black,inner sep=0pt, minimum size=3pt] (D+C) at (4.5,1.5) {};
        \draw[-] (D+S) to node[above] {D} (D+C);
        
        \node[circle,fill=black,inner sep=0pt, minimum size=3pt] (E+S) at (4,1) {};
        \node[circle,fill=black,inner sep=0pt, minimum size=3pt] (E+C) at (5.5,1) {};
        \draw[-] (E+S) to node[above] {E} (E+C);
        
        \node[circle,fill=black,inner sep=0pt, minimum size=3pt] (F+S) at (5,1.5) {};
        \node[circle,fill=black,inner sep=0pt, minimum size=3pt] (F+C) at (6.5,1.5) {};
        \draw[-] (F+S) to node[above] {F} (F+C);
        
        \node[blue] at (.5,2.5) {$\vdots$};
        \node[blue] at (1.5,2.5) {$\vdots$};
        \node[blue] at (2.5,2.5) {$\vdots$};
        \node[blue] at (3.5,2.5) {$\vdots$};
        \node[blue] at (4.5,2.5) {$\vdots$};
        \node[blue] at (5.5,2.5) {$\vdots$};
        \node[blue] at (6.5,2.5) {$\vdots$};
        \node[blue] at (7.5,2.5) {\reflectbox{$\ddots$}};
        \node[blue] at (7.5,1.5) {$\cdots$};
         \node[blue] at (7.5,1) {$\cdots$};
        %\node[circle,fill=black,inner sep=0pt, minimum size=3pt] (G+S) at (5.8,2.1) {};
        %\node[circle,fill=black,inner sep=0pt, minimum size=3pt] (G+C) at (7,2.1) {};
        %\draw[-] (G+S) to node[above] {A} (G+C);
        
        %\node[circle,fill=black,inner sep=0pt, minimum size=3pt] (H+S) at (7.5,2.6) {};
        %\node[circle,fill=black,inner sep=0pt, minimum size=3pt] (H+C) at (9,2.6) {};
        %\draw[-] (H+S) to node[above] {G} (H+C);
        
        \draw[->] (0,.5) to (8,.5);
        \node[] (t) at (8,.7) {$t$};
    \end{tikzpicture}
    }
    \caption{ Dots indicate that the shown pattern of chained activity instances can be extended arbitrarily, horizontally as well as vertically}
    \label{fig:chained_activity_instances}
    \end{subfigure}
    \hfill
    \begin{subfigure}[b]{0.24\textwidth}
        \resizebox{\textwidth}{!}{%
        \begin{tikzpicture}
          \node[](A){A};
          \node[below=of A](B){B};
          \node[right=of A](C){C};
          \node[right=of C](E){E};
          \node[right=of B](D){D};
          \node[right=of D](F){F};
          \draw[->] (A) to (D);
          \draw[->] (A) to (C);
          \draw[->,bend left,dotted] (A) to (E);
          \draw[->,dotted] (A) to (F);
          \draw[->] (B) to (D);
          \draw[->] (B) to (E);
          \draw[->] (C) to (E);
          \draw[->] (C) to (F);
          \draw[->] (D) to (F);
          \draw[->,bend right,dotted] (B) to (F);
        \end{tikzpicture}
        }
        \caption{Corresponding interval order}
        \label{fig:partial_order_no_cuts}
    \end{subfigure}
    \hfill
    \begin{subfigure}[b]{0.19\textwidth}
        \centering
        \begin{tikzpicture}
          \tikzset{
            arrow/.style={
              draw=lightgray,
              minimum height=.4cm,
              inner sep=.2em,
              shape=signal,
              signal from=west,
              signal to=east,
              signal pointer angle=150,
            }
          }
          %\node[arrow] {first};
          \node[arrow,align=center,fill=black!5,text=black,inner sep=3pt,] (A) {A\\B\\C\\D\\E\\F};
          %\node[arrow, below=.1cm of A,fill=orange,text=black] (B) {B};
          %\node[arrow, below=.1cm of B,fill=green1,text=black] (C) {C};
          %\node[arrow, below=.1cm of C,fill=blue1,text=white] (D) {D};
          %\node[arrow, below=.1cm of D,fill=yellow] (E) {E};
          %\node[arrow, below=.1cm of E,fill=green2] (F) {F};

        \end{tikzpicture}
        \caption{Corresponding visualization}
        \label{fig:stacked_chevrons}
    \end{subfigure}
    \caption{Example trace and interval order in which no cuts are applicable}
\end{figure}
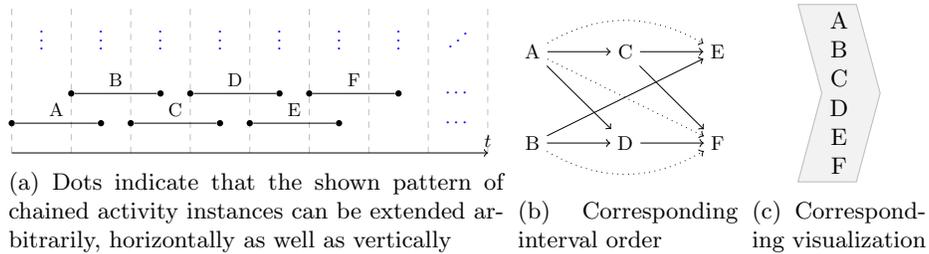

\vspace{-.4cm}

Consider \autoref{fig:chained_activity_instances}, showing an example of a trace where no cuts can be applied.
Since each activity instance is overlapping with some other activity instance, we cannot apply an ordering cut.
Also, since there is no activity instance that overlaps with all other activity instances, we cannot apply a parallel cut.
Note that the visualized pattern of chained activity instances can be arbitrarily extended by adding more activity instances vertically and horizontally, indicated by the dots in \autoref{fig:chained_activity_instances}.
\autoref{fig:partial_order_no_cuts} shows the corresponding interval order.

For the example trace, the proposed approach visualizes the activities $A,\dots,F$ within a single chevron, indicating that the activities are executed in an unspecified order (cf. \autoref{fig:stacked_chevrons}).
Thus, the visualization highly simplifies the observed process behavior in such cases.
Alternatively, it would be conceivable to show the interval order within a chevron if an (induced) interval order cannot be cut anymore.
However, we decided to keep the visualization simple and show all activities within a single chevron.
Note that this design decision entails that the expressiveness of the proposed visualization is lower than the graphical notation of interval orders, i.e., different interval orders can have the same visualization.
%In~\autoref{sec:evaluation} we will analyze this 
%Nevertheless, this limitation should not be seen as a disadvantage, since the purpose of trace visualization is to identify patterns in process execution.

\subsection{Implementation}
\label{sec:implementation}

\begin{figure}[tb]
    \centering
    \includegraphics[width=\textwidth]{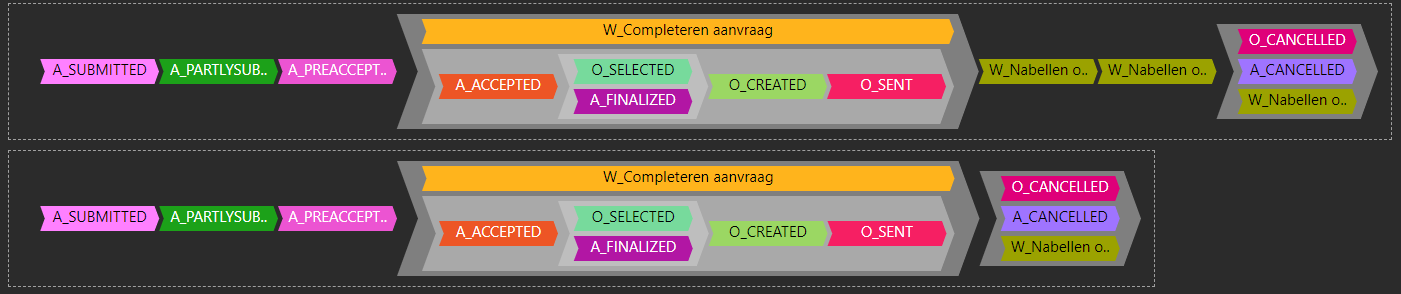}
    \caption{Screenshot of Cortado's variant explorer showing real-life event data~\cite{bpi12}}
    \label{fig:cortado}
\end{figure}

The proposed visualization approach for partially ordered event data has been implemented in \emph{Cortado}~\cite{cortado}\footnote{Available from version 1.2.0, downloadable from: \url{https://cortado.fit.fraunhofer.de/}}, which is a standalone tool for interactive process discovery.
\autoref{fig:cortado} shows a screenshot of Cortado visualizing an event log with partially ordered events.
The implemented trace variant explorer works for both, partially and totally ordered event data. 
The tool assumes an event log in the \texttt{.xes} format as input.
If the provided event log contains start- and complete-timestamps, the visualization approach presented in this paper is applied.

%The algorithms used by Cortado are written in Python programming language.
%We use the library NetworkX~\cite{hagberg2008exploring} for the computation of the visualizations, i.e., recursively applying cuts.
%The visualization in the GUI, we use web technologies. 

\section{Evaluation}
\label{sec:evaluation}
In this section, we evaluate the proposed visualization approach. 
We focus thereby on the performance aspects of the proposed visualization.
Further, we focus on the limitations, i.e., no cuts can be applied anymore, although the (induced) interval order has more than one element, as discussed in \autoref{sec:limitations}.

\begin{table}[t]
    \scriptsize
    \caption{Evaluation results based on real-life event logs}
    \label{tab:results}
    \centering
    \resizebox{\textwidth}{!}{%
    \begin{tabular}{@{} c c c c c c c c c c @{}}
        \toprule
        &\multicolumn{2}{c}{\textbf{Log statistics}}&\multicolumn{4}{c}{\makecell{\textbf{Calculation time (s) of}\\\textbf{interval ordered variants}}}&\multicolumn{2}{c}{\textbf{Variants statistics}}\\\cmidrule(lr){2-3}\cmidrule(lr){4-7}\cmidrule(lr){8-10}
        
        \makecell{Event\\log} &  
        \makecell{\#cases\\(avg. \#events\\per case)}  &   
        \makecell{multiple\\timestamps\\per activity\\available} & \makecell{Total\\calculation} & \makecell{Pre-\\processing\\event\\data} & \makecell{Creating\\interval\\orders} & \makecell{Cutting\\interval\\orders} 
        & \makecell{\#classic variants\\(only start-time-\\stamp considered)} 
        & \makecell{\#interval\\ordered\\variants}
        & \makecell{\#interval\\ordered\\variants\\with limitations}\\  \midrule
        
        \makecell[l]{BPI Ch.\\2017 \cite{bpi17}} & 
        \makecell{31,509\\(${\approx} 38$)} 
        & yes 
        & ${\approx}39$ 
        & ${\approx}21.6$  
        & ${\approx}5.7$ 
        & ${\approx}11.3$ 
        & 15,930 
        & 5,854
        & 335 (${\approx} 6\%$)\\ \\
        
        \makecell[l]{BPI Ch.\\2012 \cite{bpi12}} 
        & \makecell{13,087\\(${\approx} 20$)} 
        & yes
        & ${\approx}22.6$ 
        & ${\approx}5.9$ 
        & ${\approx}5.4$ 
        & ${\approx}11.2$ 
        & 4,366 
        & 3,830
        & 0 (${\approx} 0\%$)\\ \\
        
        \makecell[l]{Sepsis~\cite{sepsis}} 
        & \makecell{1,050\\(${\approx} 14$)} 
        & no
        & ${\approx}1.9$ 
        & ${\approx}0.3$ 
        & ${\approx}0.3$ 
        & ${\approx}1.2$ 
        & 846 
        & 690
        & 0 (${\approx} 0\%$) \\ 
        
        \bottomrule
    \end{tabular}
    }
\end{table}

We use publicly available, real-life event logs~\cite{bpi12,bpi17,sepsis}.
\autoref{tab:results} shows the results.
The first three columns show information about the logs.
Two logs~\cite{bpi17,bpi12} contain start- and complete-timestamps per activity instance while one log~\cite{sepsis} contains only a single timestamp per activity instance.
Regarding the total calculation time, we note that the duration of the visualization calculation is reasonable from a practical point of view.
%The pre-processing time includes filtering and the detection of activity instances, i.e., the mapping of start and complete events to a single activity instance.
We observe that the recursive application of cuts takes up most of the computation time in all logs, as expected.
Regarding the variants, we observe that the number of classic variants is higher compared to the number of variants derived from the interval order for all event logs.
We observe this even for the third event log~\cite{sepsis} because some activities within the cases share the same timestamp.
Regarding the limitations of the approach, as discussed in \autoref{sec:limitations}, we observe that only in the first log~\cite{bpi17} approximately in 6\% of all trace variants patterns occur where it was not possible to apply cuts anymore. 
Note that the limitation cannot occur in event logs where only a single timestamp per activity is available, e.g., \cite{sepsis}.

%Interestingly, in none of the used event logs we faced trace variants where the pattern of chained activity instances occurred, as discussed in \autoref{sec:limitations}.
%For each trace variant in all logs, we ended up with induced interval orders containing only a single element.

\section{Conclusion}
\label{sec:conclusion}
This paper introduced a novel visualization approach for partially ordered event data. 
Based on chevrons, known from the classic trace variant explorer, our approach visualizes the ordering relations between process instances in a hierarchical manner.
Our visualization allows to easily identify common patterns in trace variants from partially ordered event data.
The approach has been implemented in the tool \emph{Cortado} and has been evaluated on real-life event logs.
%For future work, we plan to investigate 

\bibliographystyle{splncs04}
\bibliography{main}

\begin{thebibliography}{10}
\providecommand{\url}[1]{\texttt{#1}}
\providecommand{\urlprefix}{URL }
\providecommand{\doi}[1]{https://doi.org/#1}

\bibitem{vanderAalst2016}
van~der Aalst, W.M.P.: Data Science in Action. Springer Berlin Heidelberg
  (2016)

\bibitem{van2021mayItakeYourOrder}
van~der Aalst, W.M.P., Santos, L.: May i take your order? on the interplay
  between time and order in process mining. arXiv preprint arXiv:2107.03937
  (2021)

\bibitem{10.2312:eurova.20151106}
Bodesinsky, P., Alsallakh, B., Gschwandtner, T., Miksch, S.: {Exploration and
  Assessment of Event Data}. In: EuroVis Workshop on Visual Analytics (EuroVA).
  The Eurographics Association (2015)

\bibitem{bpi12}
van Dongen, B.: {BPI Challenge 2012} (2012),
  \url{https://data.4tu.nl/articles/dataset/BPI_Challenge_2012/12689204}

\bibitem{bpi17}
van Dongen, B.: {BPI Challenge 2017} (2017),
  \url{https://data.4tu.nl/articles/dataset/BPI_Challenge_2017/12696884}

\bibitem{10.1007/978-3-319-19069-3_19}
van Eck, M.L., Lu, X., Leemans, S.J.J., van~der Aalst, W.M.P.: {PM}$^2$: A
  process mining project methodology. In: Advanced Information Systems
  Engineering. Springer (2015)

\bibitem{FISHBURN1970144}
Fishburn, P.C.: Intransitive indifference with unequal indifference intervals.
  Journal of Mathematical Psychology  \textbf{7}(1) (1970)

\bibitem{10.1007/978-3-319-53435-0_7}
Gschwandtner, T.: Visual analytics meets process mining: Challenges and
  opportunities. In: Data-Driven Process Discovery and Analysis. Springer
  (2017)

\bibitem{10.1007/978-3-319-42887-1_17}
Leemans, S.J.J., Fahland, D., van~der Aalst, W.M.P.: Using life cycle
  information in process discovery. In: Business Process Management Workshops.
  Springer (2016)

\bibitem{10.1007/978-3-319-15895-2_7}
Lu, X., Fahland, D., van~der Aalst, W.M.P.: Conformance checking based on
  partially ordered event data. In: Business Process Management Workshops.
  Springer (2015)

\bibitem{sepsis}
Mannhardt, F.: Sepsis cases - event log (2016),
  \url{https://data.4tu.nl/articles/dataset/Sepsis_Cases_-_Event_Log/12707639}

\bibitem{cortado}
Schuster, D., van Zelst, S.J., van~der Aalst, W.M.P.: Cortado---an interactive
  tool for data-driven process discovery and modeling. In: Application and
  Theory of Petri Nets and Concurrency. Springer (2021)

\end{thebibliography}

\end{document}